\documentclass{siamonline181217}

\usepackage{pgfplotstable,filecontents}
\pgfplotsset{compat=1.9}

\usepackage{lipsum}
\usepackage{amsopn}
\usepackage{mathtools}

\usepackage{macros-hao}
\ifpdf
  \DeclareGraphicsExtensions{.eps,.pdf,.png,.jpg}
\else
  \DeclareGraphicsExtensions{.eps}
\fi

\newcommand{\cD}{\mathcal{D}}
\newcommand{\cW}{\mathcal{W}}
\newcommand{\cT}{\mathcal{T}}

\setlist[enumerate]{leftmargin=.5in}
\setlist[itemize]{leftmargin=.5in}

\newsiamremark{remark}{Remark}
\newsiamremark{hypothesis}{Hypothesis}
\crefname{hypothesis}{Hypothesis}{Hypotheses}
\newsiamthm{claim}{Claim}

\headers{A simple bipartite graph projection model for clustering in networks}{Austin R. Benson, Paul Liu, Hao Yin}

\title{A simple bipartite graph projection model for clustering in networks}

\author{
Austin R. Benson
\thanks{Department of Computer Science, Cornell University, Ithaca, NY, USA 
  (\email{arb@cs.cornell.edu}).}
\and 
Paul Liu
\thanks{Department of Computer Science, Stanford University, Stanford, CA, USA 
  (\email{paul.liu@stanford.edu})}
\and 
Hao Yin
\thanks{Institute for Computational and Mathematical Engineering, Stanford University, Stanford, CA, USA 
  (\email{yinh@stanford.edu})}
}

\newcommand{\given}{\;\vert\;}

\usepackage{hyperref}
\definecolor{mylinkcolor}{RGB}{0,0,0}
\hypersetup{colorlinks,allcolors=mylinkcolor,citecolor=mylinkcolor}

\definecolor{mygreen}{RGB}{27,158,119}
\definecolor{myorange}{RGB}{217,95,2}
\definecolor{myblue}{RGB}{117,112,179}

\begin{document}

\maketitle


\begin{abstract}
  Graph datasets are frequently constructed by a projection of a bipartite
  graph, where two nodes are connected in the projection if they share a common
  neighbor in the bipartite graph; for example, a coauthorship graph is a
  projection of an author-publication bipartite graph. Analyzing the structure
  of the projected graph is common, but we do not have a good understanding of
  the consequences of the projection on such analyses. Here, we propose and
  analyze a random graph model to study what properties we can expect from the
  projection step. Our model is based on a Chung-Lu random graph for
  constructing the bipartite representation, which enables us to rigorously
  analyze the projected graph. We show that common network properties such as
  sparsity, heavy-tailed degree distributions, local clustering at nodes, the
  inverse relationship between node degree, and global transitivity can be
  explained and analyzed through this simple model. We also develop a fast
  sampling algorithm for our model, which we show is provably optimal for
  certain input distributions. Numerical simulations where model parameters come
  from real-world datasets show that much of the clustering behavior in some
  datasets can just be explained by the projection step.
\end{abstract}



\section{Networks as bipartite projections}

Networks or graphs that consist of a set of nodes and their pairwise
interactions are pervasive models throughout the sciences. Oftentimes, network
datasets are constructed by a ``projection'' of a bipartite
graph~\cite{Neal-2014-backbone,Newman-2001-random,Taudiere-2015-plants,Zhou-2007-projection};
specifically, given a bipartite graph with left and right nodes,
the \emph{one-mode projection} is a (unipartite) graph on the left nodes,
where two left nodes are connected if they share a common right node neighbor in
the bipartite graph.
In many cases, these projections are explicit in the data
construction process, such as connecting diseases associated with the same
gene~\cite{Goh-2007-diseasome}, people belonging to the same group or
team~\cite{Porter-2005-congress,Sun-2003-neighborhoods}, and ingredients
appearing in common recipes~\cite{Ahn-2011-flavor,Teng-2012-recipes}. In other
cases, the projection is more implicit. For example, the connections in a social
network often arise due to shared interests~\cite{breiger1974duality}.
Regardless, even though a bipartite graph is more expressive than its
projection, analyzing the projection still leads to valuable data
insights~\cite{Vogt-2010-drug-target,Zhang-2008-congress}, enables the use of
standard network analysis
tools~\cite{Benson-2016-hoo,Li-2017-inhomogeneous,zhou2007learning}, and can
even be used to make predictions about the bipartite graph
itself~\cite{Benson-2018-simplicial}.

For network analysis, it is paramount to know if structural properties in the
data arise from some phenomena of the system under study or are simply
consequences of a mathematical property of the graph construction
process. Random graph models can serve as null models for making such
distinctions~\cite{Fosdick-2018-configuration}.  Often, the random graph model
maintains some property of the network data (at least approximately or in
expectation) and then direct mathematical analysis of the random graph can be
used to determine whether certain structural properties will arise as a
consequence. For example, Chung and Lu showed that short average path lengths
can be a consequence of a uniform sample of a random graph with an expected
power law degree distribution~\cite{Chung-2002-distances}. 

Here, we analyze a simple random graph model that explains some properties of projected graphs.
More specifically, the random graph model is a projection of a bipartite ``Chung-Lu style" model. 
Each left and right node in the bipartite graph has a weight, and the probability of an edge is
proportional to the product of these weights. 

The simplicity of this model enables theoretical analysis of properties of the projected graph.
One fundamental property is \emph{clustering}: even in a
sparse network, there is a tendency of edges to appear in small clusters or
cliques~\cite{easley2010networks,rapoport1953spread,Watts-1998-collective}.
There are various explanations for clustering, including local evolutionary
processes~\cite{jackson2007meeting,granovetter1977strength,szabo2003structural},
hierarchical organization~\cite{ravasz2002hierarchical}, and community
structure~\cite{seshadhri2012community}. Here, we show how clustering can arise
just from bipartite projection. We derive an explicit equation for the expected
value of a probabilistic variant of the \emph{local clustering coefficient} of a node (the fraction of pairs of
neighbors of the node that are connected) as a function of its weight in the
model.

We show that local clustering decreases with the inverse of the weight, while expected
degree grows linearly with the weight, which is consistent with prior empirical
measurements~\cite{newman2003structure,seshadhri2012community}, 
mean-field analysis of models that explicitly incorporate clustering~\cite{szabo2003structural},
and certain random intersection graph models~\cite{bloznelis2017correlation}. 
Thus, the weights in the bipartite model
are a potential confounding factor for this relationship between degree and
clustering. 

In addition, using weight distributions fit from real-world bipartite graph data, 
we show that high levels of clustering and clustering levels at a given degree are often
just a consequence of bipartite projection. However, in several datasets, there is still
a gap between the clustering levels in the data and in the model.
Bipartite projection has been mentioned
informally as a reason for clustering in several
datasets~\cite{Fu-Yu-2019-stackexchange,Newman-2004-coauthorship,Opsahl-2013-clustering},
and a recent study has shown that sampling from \emph{configuration models} of
hypergraphs and projecting can also reproduce
clustering~\cite{chodrow2019configuration}. Our analysis provides theoretical
justifications and further explanations for these claims, and also shows
that the global clustering (also called transitivity) tends towards a positive
constant as the bipartite network grows large.
We also analyze a recently introduced measure of clustering called
the closure coefficient~\cite{yin2019local,yin2020directed} under our projection model
and find that the expected local closure coefficient of every node is the same,
which aligns with some prior empirical results~\cite{yin2019local}.

In addition to clustering, we analyze several properties of the bipartite random graph
and its projection. For instance, we show that if the weight distribution on the left and right nodes
follow a power law, then the degree distribution for those
nodes is also a power law in the bipartite graph;
moreover, the degrees in the projected graph will also follow a power law. 
Thus, heavy-tailed degree distributions in the projected graph can simply be a consequence of a process that creates
heavy-tailed degree distributions in the bipartite graph. Furthermore, we show
that the projected graph is sparse in the sense that, under a mild restriction
on the maximum weight, the probability of an edge between any two nodes goes to
zero as the number of nodes in the projected graph grows to infinity. Combined
with our results on clustering, our model thus provides a large class of
networks that are ``locally dense but globally sparse''~\cite{williamson2019random}.


\subsection{Preliminiaries}
\label{sec:preliminaries}
We consider networks as undirected graphs $G = (V, E)$ without self-loops 
and multi-edges. 
We use $\degree{u}$ to denote the degree of node $u$ (the number of edges incident to node $u$)
and $\triad{u}$ to denote the number of triangles (3-cliques) containing node $u$. 
A \emph{wedge} is a pair of edges that shared a common node, and the common node
is the \emph{center} of the wedge.
A statistic of primary interest is the \emph{clustering coefficient}:
\begin{definition}
\label{def:clustering}
The \emph{local clustering coefficient} of a node $u \in V$ is $\tilde{C}(u) = \frac{2 T(u)}{d(u) (d(u) - 1)}$,
\ie, the chance that a randomly chosen wedge centered at $u$ induces a triangle.

At the network level, the \emph{global clustering coefficient} $\tilde{C}_G$ is the probability that a randomly chosen wedge in the entire graph induces a triangle, \ie, 
$\tilde{C}_G = \frac{\sum_{u\in V} 2T(u)}{\sum_{u\in V} d(u)(d(u)-1)}.$
\end{definition}

A closely related measure of clustering is the conditional probability
of edge existence given the wedge structure~\cite{bloznelis2013degree,bloznelis2017correlation,deijfen2009random}. 
Specifically, we have the following analogs of the local and global clustering coefficients:
\begin{equation}   \label{Eq:CondProb_GCC}
C_G = \prob{(v,w) \in E \mid (u,v), (u,w) \in E},
\end{equation}
where all the nodes $u, v, w \in V$ are unspecified, while the local clustering coefficient is
\begin{equation}   \label{Eq:CondProb_LCC}
C(u) = \prob{(v,w) \in E \mid (u,v), (u,w) \in E},
\end{equation}
where $u$ is the specified node. In both cases, $(u, v)$ and $(u, w)$ comprise a random wedge from the graph.
In this paper, we use these slightly different definitions of clustering based on conditional edge existence, as they
are more amenable to analysis.

An alternative clustering metric is the recently proposed \emph{closure coefficient}~\cite{yin2019local,yin2020directed}.
\begin{definition}
\label{def:closure}
The local closure coefficient of a node $u \in V$ is $\tilde{H}(u) = \frac{2 T(u)}{W_h(u)}$, where $W_h(u)$ is the number of length-2 paths leaving vertex $u$. In other words, the closure coefficient is the chance that a randomly chosen 2-path emanating from $u$ induces a triangle.
\end{definition}

Analogously, the conditional probability variant of the closure coefficient is:
\begin{equation}   \label{Eq:CondProb_HCC}
H(u) = \prob{(u,w) \in E \mid (u,v), (v,w) \in E},
\end{equation}
where $u$ is the specified node.

The global closure coefficient is equal to the global clustering coefficient, as the number of 2-paths is exactly equal to the number of wedges. This is true for both the non-conditional and the conditional probability variant. In \Cref{sec:conditional-to-standard}, we show that the conditional probability definitions above correspond to a weighted average over the standard definitions of clustering and closure. Henceforth when referring to the clustering or closure coefficients, we always refer to the conditional probability variant.

Next, a graph is \emph{bipartite} if the nodes can be partitioned into two disjoint subsets
$L \sqcup R$, which we call the \emph{left} and \emph{right} nodes, and any edge is between one node from $L$ and one node from $R$.
We denote a bipartite graph by $G_b = (V_b, E_b)$ with $V_b = L \sqcup R$,
and call $L$ and $R$ the left and right side of the bipartite graph.
The number of nodes on each side is denoted by $n_L = \lvert L \rvert$ and 
$n_R = \lvert R \rvert$, and $n_b = \lvert V_b \rvert = n_L + n_R$ is the total number
or nodes. Analogously, for any node $u \in V_b$, we use $\degree[b]{u}$ as its degree.

The \emph{projection} of a bipartite graph is the primary concept we analyze.
\begin{definition}
\label{eq:bip_projection}
A projection of a bipartite graph $G_b = (L \sqcup R, E_b)$ is the
graph $G = (L, E)$, where the nodes are the left nodes of the bipartite graph
and the edges connect any two nodes in $L$ that connect to some node $r \in R$
in the bipartite graph. More formally,
\begin{equation}
E = \{ (u, v) \given u, v \in L, u \neq v, \text{ and }\; \exists z \in R \text{ for which } (u, z), (v, z) \in E_b \}.
\end{equation}
If there is more than one right node $z$ that connects to left nodes $u$ and $v$ in the bipartite
graph, the projection only creates a single edge between $u$ and $v$.
\end{definition}
Given a dataset, one can project onto the left or right nodes.
One can always permute the left and right nodes, and
we assume projection onto the left nodes $L$ for notational consistency.

Several statistical properties of the models we consider will use samples
drawn from a \emph{power law} distributions, which are prevalent in network data models~\cite{clauset2009power}.
\begin{definition}   
\label{Def:PowerLaw}
The probability density function of the power law distribution, 
parametrized by $(\alpha, w_{\min}, w_{\max})$
with $\alpha > 1$ and $0< w_{\min} < w_{\max} \leq \infty$, is
\[
f(w) = 
\begin{cases*}
C w^{-\alpha} & \text{if $w \in [w_{\min}, w_{\max}]$} \\
0 & \text{otherwise} 
\end{cases*}
\]
where $w > 0$ is any real number and $C = ({w_{\min}^{1-\alpha} - w_{\max}^{1-\alpha}})/({\alpha - 1})$ is a normalizing constant.

For a discrete power-law (or Zipfian) distribution, we restrict $w$ to integer values inside $[w_{\min}, w_{\max}]$ and adjust the normalization constant accordingly.
\end{definition}
The parameter $\alpha$ is the decay exponent of the distribution,  while $w_{\min}$ and $w_{\max}$ specify range.
For simplicity, we assume that $w_{\min} = 1$ and $w_{\max} = \Omega(1)$ throughout this paper.

When the maximum range is not specified, \ie, $w_{\max} = \infty$,
a standard result on the maximum statistics of power-law samples is the following:
\begin{lemma}[Folklore]
\label{lem:NumGroups} 
\label{lem:PLawMax}
For a discrete or continuous power-law distribution $\cD$ with parameters $(\alpha, w_{\min}=1, w_{\max} = \infty)$ 
and i.i.d.\ samples $w_1, w_2, \ldots, w_n \sim \cD$, $\expect{\max_i w_i} = n^{\frac{1}{\alpha-1}}.$
\end{lemma}

\section{Models for Bipartite Projection}\label{sec:model}

In this section we formalize our model and give some background on relevant models for projection and graph generation. 
Our model is an extension of the seminal random graph model from Chung and Lu~\cite{Chung-2002-distances}. 
The classical Chung-Lu model takes as input a weight sequence $S$, which specifies a nonnegative weight $w_u$ for each node,
and then produces an undirected edge $(u, v)$ with probability $w_uw_v / \sum_{z}w_z$.
To make sure that the probabilities are well defined, the model assumes that $\max_{u} w_u^2 \le \sum_{v}w_v$.
Along similar lines, Aksoy et al.\ introduced a Chung-Lu-style bipartite random graph model based on realizable degree sequences~\cite{aksoy2017measuring}.
In general, the model we use is quite similar. However, our focus in this paper is to analyse the effects of projection on such models.

\subsection{Our Chung-Lu Style Bipartite Projection Model}
Our model takes as input the number of left nodes $n_L$,
the number of right nodes $n_R$, and two sequences of weights $S_L$ and $S_R$ for the left and right nodes.
We denote the weight of any node $u$ by $w_u$.
The model then samples a random bipartite graph $G_b = (L \sqcup R, E_b)$, where
\begin{equation}\label{Eq:basic_ChungLu}
\prob{(u, v) \in E_b \given w_u, w_v} = \min\left(\frac{w_u w_v}{\sum_{z \in R}w_z}, 1\right),\; u \in L,\; v \in R.
\end{equation}
After generating the graph, we project the graph following \Cref{eq:bip_projection},
which is itself a random graph.
This model is similar to the inhomogeneous random intersection graph~\cite{bloznelis2016clustering}
(see \cref{sec:related_projection} for more details).

Our analysis will depend on properties of $S_L$ and $S_R$ and the moments
of these sequences. We denote the $k$th-order moments of $S_L$ and $S_R$
by $M_{Lk}$ and $M_{Rk}$ for integers $k \geq 1$:
\begin{align}
M_{Lk} = \frac{1}{n_L}\sum_{u \in L} w_u ^k,\quad
M_{Rk} = \frac{1}{n_R}\sum_{v \in R} w_v ^k.
\end{align}
With this notation, we can re-write the edge probabilities as
\begin{equation}\label{Eq:Assm_ChungLu}
\prob{(u, v) \in E_b \given w_u, w_v} = \min\left(\frac{w_u w_v}{n_RM_{R1}}, 1\right).
\end{equation}

\begin{remark}
\label{rem:equal-degrees}
The model is invariant upon uniform scalings of the weight sequence $S_R$. 
Thus we can assume without loss of generality that $n_R \expect{M_{R1}} = n_L \expect{M_{L1}}$.
This corresponds to the natural condition that the expected degree sum of the left and right side is equal.
\end{remark}

A practical concern is how efficiently we can sample from this model,
as naive sampling of the bipartite graph requires $n_Ln_R$ coin flips. 
There are fast sampling heuristics for the bipartite graph, 
based on sampling each node in an edge individually for some pre-specified number of edges~\cite{aksoy2017measuring}. 
We develop a fast sampling algorithm in \Cref{sec:fast_sampling} that has some 
theoretical optimality guarantees for sequences $S_L$ and $S_R$ with certain properties.

\subsection{Configuration models}
Much of our motivation for random graph models is that they provide a baseline
for what graph properties we might expect in network data just from a simple
underlying random process (in our case, we are particularly interested in what
graph properties we can expect from projection). In turn, this helps researchers
determine which properties of the data are interesting or inherent to the system
modeled by the graph.

While Chung-Lu models aim to preserve input degree sequences in expectation, 
\emph{configuration models} preserve degrees exactly,
sampling from the space of graphs with a specified degree sequence~\cite{Fosdick-2018-configuration}.
Configuration models for bipartite graphs have only been studied in earnest
recently~\cite{chodrow2019configuration}, where the goal is to sample bipartite
graphs with a specified degree sequence for the left and right
nodes. A bipartite configuration model inherits
many benefits of a standard configuration model; for instance, the degree
sequence is preserved exactly, creating an excellent null model for a given dataset.

At the same time, configuration models carry some restrictions.  First, the
random events on the existence of two edges are dependent (though weakly). To
see this, in a stub-labeled bipartite graph, if we condition on an edge existing
between $u \in L$ and $v \in R$, then there is one fewer stub for each node,
making them less likely to connect to other nodes. This makes theoretical
analysis difficult. Second, to generate a random graph, a configuration model
needs a degree sequence that is realizable. While the Gale--Ryser theorem
provides a simple way to check if a candidate bipartite degree sequence is
realizable~\cite{ryser1963combinatorial}, configuration models typically analyze
a given input graph rather than a class of input graphs with some property.
Third, efficient uniform sampling algorithms rely on Markov Chain Monte Carlo,
for which it is extremely difficult to obtain reasonable mixing time bounds.

The Chung-Lu approach (for either bipartite or unipartite graphs) sacrifices
control over the exact degree sequence for easier theoretical analysis while
maintaining the \emph{expected} degree sequence. Unlike the configuration model,
the existences of two distinct edges are independent events, there is no need to
specify a realizable degree sequence, and samples can be immediately
generated. In unipartite graphs, this has led to remarkable results on random
graphs with expected power-law degree sequences, such as small average node
distance and diameter~\cite{Chung-2002-distances}, the existence of a giant
connected component~\cite{chung2002connected}, and spectral
properties~\cite{chung2004spectra}.

\subsection{Related projection-based models}\label{sec:related_projection}
There are random graph models for bipartite graphs that are motivated by how the
projection step can lose information about community structure in the
data~\cite{Guimera-2007-bipartite,Larremore-2014-biSBM}. While these identify
possible issues with the projection, we are motivated by the fact that
several datasets are constructed via projection, either implicitly or
explicitly. There are also many models based on communities, where edge
probabilities depend on community
membership~\cite{Airoldi-2008-MMSBM,Karrer-2011-DCSBM,seshadhri2012community,Yang-2012-affiliation}.
These models can be interpreted as probabilistic projections of node-community
bipartite graphs. Such models are typically fit from data to reveal cluster
structure. Such analysis is not the focus of this paper.

There are a few random graph models where a random bipartite graph is
deterministically projected~\cite{barber2008clique,chodrow2019configuration,Lattanzi-2009-affiliation,williamson2019random}.
Some of these have specifically considered clustering, which is of primary
interest for us. A recent example is the configuration model for
hypergraphs~\cite{chodrow2019configuration}, which can be interpreted as a
bipartite random graph model: the nodes in the hypergraph are the left nodes in
the bipartite graph, and the right nodes in the bipartite graph correspond to
edges in the hypergraph. Chodrow~\cite{chodrow2019configuration} found that the clustering of projections of
bipartite representations of several real-world hypergraph datasets was similar
to or even less than the clustering of projections of samples from the
configuration model. Similar empirical results have been found on related
datasets, under a model that samples the degrees of the left nodes in the
bipartite graph according to a distribution learned from the data and connects
the edges to the right nodes uniformly at
random~\cite{Fu-Yu-2019-stackexchange}. Our theoretical analysis provides
additional grounding for these empircal results, and our model provides a
Chung-Lu-style alternative to the configuration model approach.

In terms of theoretical results, the models most related to ours are
\emph{random intersection graphs}~\cite{bloznelis2013degree,godehardt2003two} 
and \emph{random clique covers}~\cite{williamson2019random}. In these models, 
a graph is constructed by sampling $n$ sets from a universe of size $m$ according to a distribution $D$. 
A node is associated to each of the $n$ sets, and two vertices in the graph are adjacent if their subsets overlap. 
This is equivalent to representing the sets as an $n$-by-$m$ bipartite graph and then projecting the graph onto the left nodes. 
Such models can also produce several key properties of projected graphs in practice, 
such as power-law degree distributions and negative correlation of clustering and projected degree.
In contrast to these approaches, our model can specify degree distributions 
on both sides of the bipartite graph, as opposed to just one side.
\emph{Inhomogeneous random intersection graphs} also support 
arbitrary degree distribution on both sides~\cite{bloznelis2016clustering,bloznelis2017correlation},
and justify the negative correlation of local clustering and projected degree.
In comparison, our analysis is conducted conditional on the degree sequence,
which is potentially generated from a distribution with infinite moment,
and thus requires a weaker and more realistic assumption on the degree distribution
than results from Bloznelis and Petuchovas~\cite{bloznelis2019local,bloznelis2017correlation};
however, their results work directly with projected degrees,
which is advantageous.

\section{Theoretical Properties of the Projection Model}\label{sec:theory}
In this section we provide results for graph statistics on the projected graph,
such as the degree distribution, clustering coefficients, and closure
coefficients.  For intuition, one may think of the input weight distributions to
our model as the degree distribution of a class of input graphs. As we show
in \Cref{sec:experiment}, these input weights often follow a power law
distribution in real-world datasets. Due to the simplicity of our model, it is
possible to derive analytical expressions when the input weight distribution
follows a power law (\Cref{Def:PowerLaw}).

At a high-level, for a broad range of weight distributions (including power-law distributions),
the projected graph has the following properties.
\begin{enumerate}
\item The projected graph is sparse (edge probabilities go to zero).
\item Expected local clustering at a node decays with the node's weight, and the node's weight is directly proportional to its degree in expectation.
\item Expected local closure at a node is the same for all nodes.
\item Global clustering and closure (transitivity) is a positive constant.
In other words, clustering does not go to zero as the graph grows large.
\end{enumerate}
Besides theoretical analysis, we also verify some key results with simulations, 
which relies on a fast sampling algorithm that we develop in \cref{sec:fast_sampling}.

\subsection{Assumptions on Weight Sequences}
Our analysis is conditional on the general input weight sequence on both sides of the bipartite graph.
We first assume that the normalized product of weights is at most one, making the edge existence probability (\Cref{Eq:basic_ChungLu}).
\begin{assumption}[Well-defined probabilities]
\label{assm:no_min}
  The weight in the sequences $S_L$ and $S_R$ satisfy $\frac{w_u w_v}{n_R M_{R1}} \leq 1$
  for any nodes $u \in L$, $v \in R$. 
\end{assumption}

Moreover, our analysis is asymptotic, meaning that the result holds
with high accuracy on large networks, i.e., $n_L, n_R \to \infty$.
For any two quantities $f$ and $g$, we use the following big-$O$ notations 
in the limit of $n_L, n_R \to \infty$: $f = o(g)$ if $f/g \to 0$; $f = O(g)$ if $f/g$ is bounded; 
and $f = \Omega(g)$ if $f/g$ is bounded away from 0.
We make the following assumption on the range and moment of weight sequences.
\begin{assumption}[Bounded weight sequences]
\label{assm:WeightsSeq}
\label{assm:WeightSeqMoments}
There exists a constant $\delta > 0$ such that
\begin{itemize}
\item \textnormal{(bounded range)}\;\; $\max[S_L, S_R] = O\left(n_R^{1/2 - \delta}\right)$, $\min[S_L] = \Omega(1)$ and
\item \textnormal{(bounded $S_R$ moments)}\;\;  $M_{R2} = O(M_{R1}^2)$, $M_{R4} = O\left(n_R^{1 - 2\delta}\right)$,
\end{itemize}
as $n_L, n_R \rightarrow \infty$.
\end{assumption}

\Cref{assm:WeightSeqMoments} actually specifies a family of assumptions parameterized by $\delta$,
with larger $\delta$ imposing stronger assumptions. Unless otherwise stated, we only require that $\delta > 0$. 
In the theoretical analysis of clustering coefficient, we sometimes require $\delta > 1/10$.

We do not assume the rate of which $n_L, n_R \to \infty$,
or any direct relationships between $n_L$ and $n_R$.
This makes our assumptions weaker than a wide range of assumptions
typical in the literature, such as having $n_L = \beta n_R^\sigma$
for certain $\beta, \sigma > 0$~\cite{bloznelis2013degree,bloznelis2017correlation,deijfen2009random,williamson2019random}.

Before presenting the properties of bipartite or projected graph under these
assumptions, we first show that these assumptions are naturally satisfied
if the weight sequences are generated from the power-law distribution. 
\begin{proposition}   \label{Prp:AssumptionSatisfied}
If the sequences $S_L$ and $S_R$ are independent generated from the power-law distribution
with $w_{\max} = n_R^{1/2 -\delta}$, and the right side distribution has decay exponent 
$\alpha_R > 3$, then \Cref{assm:WeightSeqMoments} is satisfied.
\end{proposition}
\begin{proof}
The bounded range requirement is automatically satisfied due to max capping,
and we focus on the bounded moment requirement.

Let $W$ be the random variable denoting a sample weight in $S_R$. 
Since $M_{R1} \geq 1$, due to the law of large numbers, it suffices to show that
$\expect{W^2} < \infty$ and $\expect{W^4} = O(n_R ^{1 - 2\delta})$ as $n_R \rightarrow \infty$.
The first result can be easily verified, and when $\alpha_R \geq 5$,
a straight-forward computation shows that $\expect{W^4} = O(\log n_R)$.
When $\alpha_R \in (3, 5)$, we have
\[
\expect{W^4} 
= \int_{1} ^{n_R ^{1/2 - \delta}} C_{\alpha, \delta} \cdot w^{-\alpha + 4}~\mathrm d w
= \frac{C_{\alpha, \delta}}{(5-\alpha)} \left(n_R ^{({5-\alpha})(\frac{1}2-\delta)} - 1 \right)
= O\left(n_R ^{1-2\delta}\right),
\]
where $C_{\alpha, \delta} =  (1 - n_R^{(1-\alpha)(1/2-\delta)}) / (\alpha - 1) = O(1)$ is the normalizing constant.
\end{proof}

Therefore, \Cref{assm:WeightSeqMoments} is
satisfied when the weight sequences are generated from power-law
distributions with only a mild requirement on the decay exponent.
In contrast, some results require constant weights on the right side~\cite{bloznelis2013degree,deijfen2009random}
or $\alpha_R > 5$ (for a finite fourth-order moment)~\cite{bloznelis2017correlation}.

When $M_{R1} \geq 1$, \cref{assm:no_min} is a direct consequence of
\Cref{assm:WeightSeqMoments} for large graphs since $\max[S_L, S_R] = o(\sqrt{n_R})$. 
Henceforth, for our theoretical analysis,
we assume that both \Cref{assm:no_min} and \Cref{assm:WeightsSeq} are satisfied.

As a final note, a direct consequence of \Cref{assm:WeightSeqMoments} is that
$\prob{(u, v) \in E_b \mid w_u, w_v} \rightarrow 0$
due to $w_u, w_v = o(n_R)$, meaning that the bipartite network is sparse. 


\subsection{Degree distribution in the bipartite graph}
In this section, we study the degree distribution in the bipartite graph with respect to a given input weight distribution.
\begin{theorem}\label{thm:degree_bip}
For any node $u \in L$, conditional on $u$'s weight $w_u$, the bipartite degree $\degree[b]{u}$ of $u$ converges in distribution to a Poisson random variable with mean $w_u$ as $n_R \rightarrow \infty$. Analogously, for any $v \in R$, conditional on $w_v$, $\degree[b]{v}$ converges in distribution to a Poisson random variable with mean $w_v$ as $n_L \rightarrow \infty$.
\end{theorem}

\begin{proof}
By symmetry, we just need to prove the result for a node $u \in L$.
For any $v \in R$, the indicator function $\indic{(u, v)\in E_b}$ is a Bernoulli 
random variable with positive probability $\frac{w_u w_v}{n_R M_{R1}}$. By a Taylor expansion,
its characteristic function can be written as
\[
\phi_{uv}(t) = 1 + (e^{it} - 1) \frac{w_u w_v}{n_R M_{R1}}
=e^{\frac{w_u w_v}{n_R M_{R1}}  (e^{it} - 1) \cdot (1 + o(1))},
\]
where the $o(1)$ term comes from the bounded range condition in \Cref{assm:WeightsSeq}.
The bipartite degree of node $u$ is the sum of the indicator functions of all
nodes $v \in R$, which are independent random variables. Thus, its
characteristic function of $\degree[b]{u}$ can be written as
\[
\phi_{\degree[b]{u}}(t) = \prod_{v \in R}\phi_{uv}(t) 
= e^{w_u \frac{\sum_{v \in R} w_v}{n_R M_{R1}} (e^{it} - 1) \cdot (1 + o(1))}
\rightarrow e^{w_u (e^{it} - 1)}.
\]
The limiting characteristic function is the characteristic function of a Poisson
random variable with mean $w_u$. Thus, $\degree[b]{u}$ converges in distribution
to a Poisson random variable with mean $w_u$ by L\'evy's continuity theorem.
\end{proof}

One corollary of \Cref{thm:degree_bip} is that, in the limit, the expected degree
of any node $u$ is its weight $w_u$, which provides an interpretation of the
node weights. Next, we show that if the weights are independently generated from
a power-law distribution, then the degrees in the bipartite graph are power-law distributed as well.

\begin{theorem}   \label{thm:degree_bip_unconditonal}
Suppose that the node weights on the left are independently sampled from a continuous power-law distribution with exponent $\alpha_L$.
Then, for any node $u \in L$, as $n_R \rightarrow \infty$, 
we have that $\prob{d_b(u) = k} \propto k^{-\alpha_L}$ for large $k$.

Similarly, suppose that the node weights on the right are independently sampled from a continuous power-law distribution with exponent $\alpha_R$. 
Then, for any node $v \in R$, as $n_L \rightarrow \infty$, 
we have that $\prob{d_b(u) = k} \propto k^{-\alpha_R}$ for large $k$.
\end{theorem}

\begin{proof}
Again, by symmetry, we only need to show the result for a node on the left.
For any node $u \in L$, according to \Cref{thm:degree_bip}, its bipartite degree distribution 
converges to a Poisson distribution with mean $w_u$. For any integer $k > \alpha$,
\begin{eqnarray*}
\prob{\degree[b]{u} = k}
  &=& \int_{1} ^{w_{\max}} \prob{\degree[b]{u} = k \mid w_u = w}\cdot f_L(w) ~ \mathrm d w\\
&=& C \int_{1} ^{w_{\max}} e^{-w} \frac{w^k}{k!} \cdot w^{-\alpha_L} ~ \mathrm d w\\
&=& \frac{C}{k!} \left( \int_{0} ^{\infty} e^{-w} w^{k-\alpha_L} ~ \mathrm d w 
- \int_{0} ^{1} e^{-w} w^{k-\alpha_L} ~ \mathrm d w
- \int_{w_{\max}} ^{\infty} e^{-w} w^{k-\alpha_L} ~ \mathrm d w  \right)\\
&=& \frac{C}{\Gamma(k+1)} (\Gamma(k - \alpha_L + 1) - O(1)) \rightarrow C k^{-\alpha_L} (1 + o(1)).
\end{eqnarray*}
Here, $C$ is a normalizing constant constant. 
The second to last line is due to the fact that $w_{\max} = \Omega(1)$, and
the last line follows because $\Gamma(k - \alpha_L + 1) / \Gamma(k + 1) \rightarrow k^{-\alpha_L}$ as $k \rightarrow \infty$.
\end{proof}


\subsection{Edge density and degree distribution in the projected graph}
To study the edge density and degree distribution in the projected graph,
we use the following quantity:
\begin{equation}
p_{u_1 u_2} := \frac{M_{R2}}{M_{R1} ^2} \frac{w_{u_1} w_{u_2}}{n_R}.
\end{equation}
The following theorem shows that $p_{u_1 u_2}$ is the asymptotic edge existence probability between the two nodes $u_1$ and $u_2$ in the \emph{projected graph}.
Note that under \Cref{assm:WeightSeqMoments}, we have $w_{u_1}, w_{u_2} = O(n_R^{1/2-\delta})$ 
and thus $p_{u_1 u_2} = O(n_R ^{-2\delta}) = o(1)$, so the projected graph is sparse as the number of nodes goes to infinity.
\begin{theorem}   \label{thm:density_proj}
For any $u_1, u_2 \in L$, as $n_R \to \infty$, we have
\[
\prob{(u_1, u_2) \in E  \mid S_L, S_R} =
p_{u_1 u_2} - \frac{p_{u_1 u_2} ^2}{2} + \left(\frac{p_{u_1 u_2}}{6} + \frac{M_{R4}}{2n_RM_{R2}^2} \right) p_{u_1 u_2}^2 \cdot (1+O(n_R^{-2\delta})).
\]
\end{theorem}
\begin{proof}
We consider the complementary case when $u_1$ and $u_2$ are not connected in 
the projected graph. This is the case when, for any nodes $v \in R$, it is connected 
to at most one of $u_1$ and $u_2$ in the bipartite graph. For each single node $v \in R$, this case happens with probability 
$1 - \frac{w_{u_1} w_{u_2} w_v ^2}{n_R ^2 M_{R1} ^2}$.
Therefore,
\begin{eqnarray*}
& \log (\prob{(u_1,u_2) \notin E \mid S_L, S_R} )
  = \sum_{v \in R} \log\left( 1 - \frac{w_{u_1} w_{u_2} w_v ^2}{n_R ^2 M_{R1} ^2} \right)
\\
&= \sum_{v \in R} \left[ - \frac{w_{u_1} w_{u_2} w_v ^2}{n_R ^2 M_{R1} ^2} - \frac{w_{u_1}^2 w_{u_2}^2 w_v ^4}{2 n_R ^4 M_{R1} ^4} \cdot (1+O(n_R^{-4\delta}))\right]
= - p_{u_1 u_2} -   \frac{M_{R4}}{2n_RM_{R2}^2} p_{u_1 u_2}^2  \cdot (1+O(n_R^{-4\delta})).
\end{eqnarray*}
Consequently,
\[
\prob{(u_1,u_2) \in E \mid S_L, S_R} 
=  p_{u_1 u_2} - \frac{p_{u_1 u_2} ^2}{2} + \left(\frac{p_{u_1 u_2}}{6} + \frac{M_{R4}}{2n_RM_{R2}^2} \right) p_{u_1 u_2}^2 \cdot (1+O(n_R^{-2\delta})).
\]
\end{proof}

We now examine the expected degree distribution of the projected graph. 
One concern is the possibility of multi-edges in our definition of a projection, 
which occurs when two nodes $u_1, u_2 \in L$ have more than one common neighbor in the bipartite graph. 
The following lemma shows that the probability of having multi-edges
conditional on edge existence is negligible, 
meaning that we can ignore the case of multi-edges with high probability.
\begin{lemma}   \label{lem:multiedge}
Let $u_1, u_2 \in L$, and let $N_{u_1 u_2}$ be the number of common neighbors of $u_1$ and $u_2$
in the bipartite graph, then
$\prob{N_{u_1 u_2} \geq 2 \mid S_L, S_R, (u_1, u_2) \in E} = O(p_{u_1 u_2})$
as $n_R \to \infty$.
\end{lemma}
\begin{proof}
Note that it suffices to show that $\prob{N_{u_1 u_2} \geq 2 \mid S_L, S_R} = O(p_{u_1 u_2}^2)$.
By the tail formula for expected values,
\begin{align*}
\expect{N_{u_1 u_2} \mid S_L, S_R} 
  &= \textstyle \sum_{k =1} ^\infty k \cdot \prob{N_{u_1 u_2} = k \mid S_L, S_R} \\
  &\textstyle \geq 2 \cdot \prob{N_{u_1 u_2} \ge 2 \mid S_L, S_R} + \prob{N_{u_1 u_2} = 1 \mid S_L, S_R} \\
  &\textstyle= \prob{N_{u_1 u_2} \ge 2 \mid S_L, S_R} + \prob{N_{u_1 u_2} \ge 1 \mid S_L, S_R}.
\end{align*}
Note that we also have
\[
\textstyle \expect{N_{u_1 u_2} \mid S_L, S_R} = \sum_{v \in R} \prob{(u_1, v), (u_2, v) \in E_b \mid S_L, S_R} = p_{u_1 u_2},
\]
and consequently
\[
\textstyle \prob{N_{u_1 u_2} \geq 2 \mid S_L, S_R} \leq p_{u_1 u_2} - \prob{N_{u_1 u_2} \geq 1 \mid S_L, S_R} \leq \frac{1}{2} p_{u_1 u_2}^2 + o(p_{u_1 u_2} ^2).
\]
The inequality uses the fact that the event $N_{u_1, u_2} \geq 1$ is equivalent to the existence of edge $(u_1, u_2)$
in the projected graph, which happens with probability $p_{u_1 u_2} - \frac{1}{2} \cdot p_{u_1
u_2}^2 + o(p_{u_1 u_2} ^2)$ by \Cref{thm:density_proj}. 
\end{proof}

Now we are ready to analyze the degree of a node in the projected graph.
The following theorem says that degree of a node in the projected is directly
proportional the weight of the node. Thus, at least in expectation, we can think
of the weight as a proxy for degree.
\begin{theorem}   \label{thm:deg_proj_mean}
For any $u \in L$, as $n_L, n_R \to \infty$, we have
\[
\expect{ \degree{u} \mid S_L, S_R} = 
\frac{M_{R2} M_{L1}}{M_{R1} ^2} \cdot \frac{n_L}{n_R} \cdot w_u \cdot (1 + o(1)),
\]
\end{theorem}

\begin{proof}
By \Cref{thm:density_proj},
\begin{align*}
\expect{\degree{u} \mid S_L, S_R} 
  = \sum_{u_1 \in L, u_1 \neq u} \prob{(u, u_1) \in E \mid S_L, S_R}
  &= \sum_{u_1 \in L, u_1 \neq u} \frac{w_u w_{u_1}}{n_R} \cdot \frac{M_{R2}}{M_{R1}^2} \cdot (1 + o(1)) \\
  &= \frac{M_{R2} M_{L1}}{M_{R1} ^2} \cdot \frac{n_L}{n_R} \cdot w_u \cdot (1 + o(1)).
\end{align*}
\end{proof}

By \Cref{thm:degree_bip_unconditonal}, the bipartite 
degree distributions of the left and right nodes are power-law distributions with exponents $\alpha_L$ and $\alpha_R$.
For such bipartite graphs, Nacher and Aktsu~\cite{nacher2011degree} showed that
the degree sequence of the projected graph follows a power law distribution.
\begin{corollary}[Section 2, \cite{nacher2011degree}]   \label{thm:deg_proj_distri}
Suppose the node weights on the left and right follow power-law distributions with exponents
$\alpha_L$ and $\alpha_R$. Then the degree distribution of the projected graph is a power-law distribution with decay exponent
$\min(\alpha_L, \alpha_R - 1)$.
\end{corollary}

When $\alpha_R \in (3, 4)$, \Cref{assm:WeightSeqMoments} is satisfied by \Cref{Prp:AssumptionSatisfied}, and
the projected graph would have power-law degree distribution with decay exponent within $(2, 3)$, which is
a standard range for classical theoretical models~\cite{dorogovtsev2002evolution} and is also observed in real-world data~\cite{broido2019scale}. We estimate $\alpha_R \in (3, 4)$ for several real-world bipartite networks
that we analzye (\cref{sec:power-law-stats}).

\subsection{Clustering in the projected graph}
In this section we compute the expected value of the clustering and closure coefficients. 
\Cref{thm:lccf_wt} rigorously analyzes the expected value of
local clustering coefficients on networks generated from projections of general bipartite random graphs. 
Our results show how (for a broad class of random graphs)
the expected local clustering coefficient varies with the node weight: 
it decays at a slower rate for small weight and then decays as the inverse of the weight for large weights.
Combined with the result that the expected projected degree is proportional to the node weight (\Cref{thm:deg_proj_mean}), 
this says that there is an inverse correlation of node degree with the local clustering coefficient, which we also verify with simulation.
This has long been a noted empirical property of complex networks~\cite{newman2003structure}, and our analysis
provides theoretical grounding, along with other recent results~\cite{bloznelis2019local,bloznelis2017correlation}.

\begin{theorem}   \label{thm:lccf_wt}
If \Cref{assm:WeightSeqMoments} is satisfied with $\delta > \frac{1}{10}$,
then conditioned on $S_L$ and $S_R$ for any node $u \in L$, we have in the projected graph that
\[
\lccf{u} =\frac{1}{1 + \frac{M_{R2} ^2}{M_{R3} M_{R1}} w_u} + o(1).
\]
\end{theorem}
Besides the trend of how local clustering coefficient decays with node weight,
we highlight how the sequence moment of $S_R$ influences the clustering coefficient.
If the distribution of $S_R$ has a heavier tail, then $\frac{M_{R2} ^2}{M_{R3} M_{R1}}$ is small (via Cauchy-Schwartz),
and one would expect higher local clustering compared to cases where
$S_R$ is light-tailed~\cite{bloznelis2017correlation} or
uniform~\cite{bloznelis2013degree, deijfen2009random}.
We also observe this higher level of clustering in simulations (\Cref{fig:cc-vs-degree}).

We break the proof of \Cref{thm:lccf_wt} into several lemmas. 
From this point on, we assume $\delta > 1/10$.
We first present the following results on the limiting probability of wedge and triangle existence,
with proofs given in \Cref{sec:proofs}.
\begin{lemma}\label{lem:wedge_prob}
As $n_R \rightarrow \infty$, for any node triple $(u_1, u, u_2)$,
the probability that they form a wedge centered at $u$ is
\begin{eqnarray*}
 \prob{(u, u_1), (u, u_2) \in E \mid S_L, S_R} 
&=& \left(1 +  \frac{M_{R1} M_{R3}}{M_{R2}^2} \cdot \frac{1}{w_u} \right)  p_{u u_1} p_{u u_2} \cdot (1 + o(1)).
\end{eqnarray*}
\end{lemma}

\begin{lemma}\label{lem:triangle_prob}
In the limit of $n_R \rightarrow \infty$, the probability of a node triple
$(u_1, u, u_2)$ forms a triangle is
\begin{eqnarray*}
\prob{(u, u_1) , (u, u_2) , (u_1, u_2) \in E \mid S_L, S_R}
= p_{u u_1} p_{u u_2} \cdot \frac{M_{R1}M_{R3}}{M_{R2}^2}\cdot \frac{1}{w_u} \cdot (1 + o(1))
+ o(p_{u u_1} p_{u u_2}).
\end{eqnarray*}
\end{lemma}

Now we have the following key result on the conditional probability triadic closure.

\begin{lemma}\label{lem:wedge_close_prob}
In the limit of $n_L, n_R \rightarrow \infty$, if a node triple $(u_1, u, u_2)$ 
forms an wedge, then the probability of this wedge being closed is
\[
\prob{ (u_1, u_2) \in E \mid S_L, S_R, (u, u_1), (u, u_2) \in E}=
\frac{1}{1 + \frac{M_{R2} ^2}{M_{R3} M_{R1}} w_u} + o(1).
\]
\end{lemma}
\begin{proof}
By combining the result of \Cref{lem:wedge_prob,lem:triangle_prob}, we have
\begin{align*}
& \prob{ (u_1, u_2) \in E \mid S_L, S_R, (u, u_1), (u, u_2) \in E} \textstyle
= \frac{\prob{(u, u_1) , (u, u_2) , (u_1, u_2) \in E \mid S_L, S_R}}
{\prob{(u, u_1), (u, u_2) \in E \mid S_L, S_R}} 
\\ & = \textstyle 
\frac{p_{u u_1} p_{u u_2} \frac{M_{R1}M_{R3}}{M_{R2}^2}\cdot \frac{1}{w_u} \cdot (1+o(1)) + o(p_{u u_1} p_{u u_2})}
{\left(\frac{M_{R1}M_{R3}}{M_{R2} ^2} \cdot \frac{1}{w_u} + 1\right) p_{u u_1} p_{u u_2} \cdot (1 + o(1))}
= \frac{1 + o(1)}{1 + \frac{M_{R2} ^2}{M_{R3} M_{R1}} w_u} + o(1). 
\end{align*}
\end{proof}

Finally, we are ready to prove our main result.
\begin{proof}[Proof of \Cref{thm:lccf_wt}]
According to \Cref{Eq:CondProb_LCC}, the local clustering coefficient is 
the conditional probability that a randomly chosen wedge centered at node $u$ forms a triangle. 
\Cref{lem:wedge_close_prob} shows that this probability is asymptotically the same 
regardless of the weights on the wedge endpoints $u_1, u_2$. 
Therefore conditioned on $S_L$ and $S_R$, we have
\begin{align*}
\lccf{u} = 
\prob{ (u_1, u_2) \in E \mid S_L, S_R, (u, u_1), (u, u_2) \in E}=
\frac{1}{1 + \frac{M_{R2} ^2}{M_{R3} M_{R1}} w_u} + o(1).
\end{align*}
\end{proof}

\begin{figure}[t]
\includegraphics[width=0.49\textwidth]{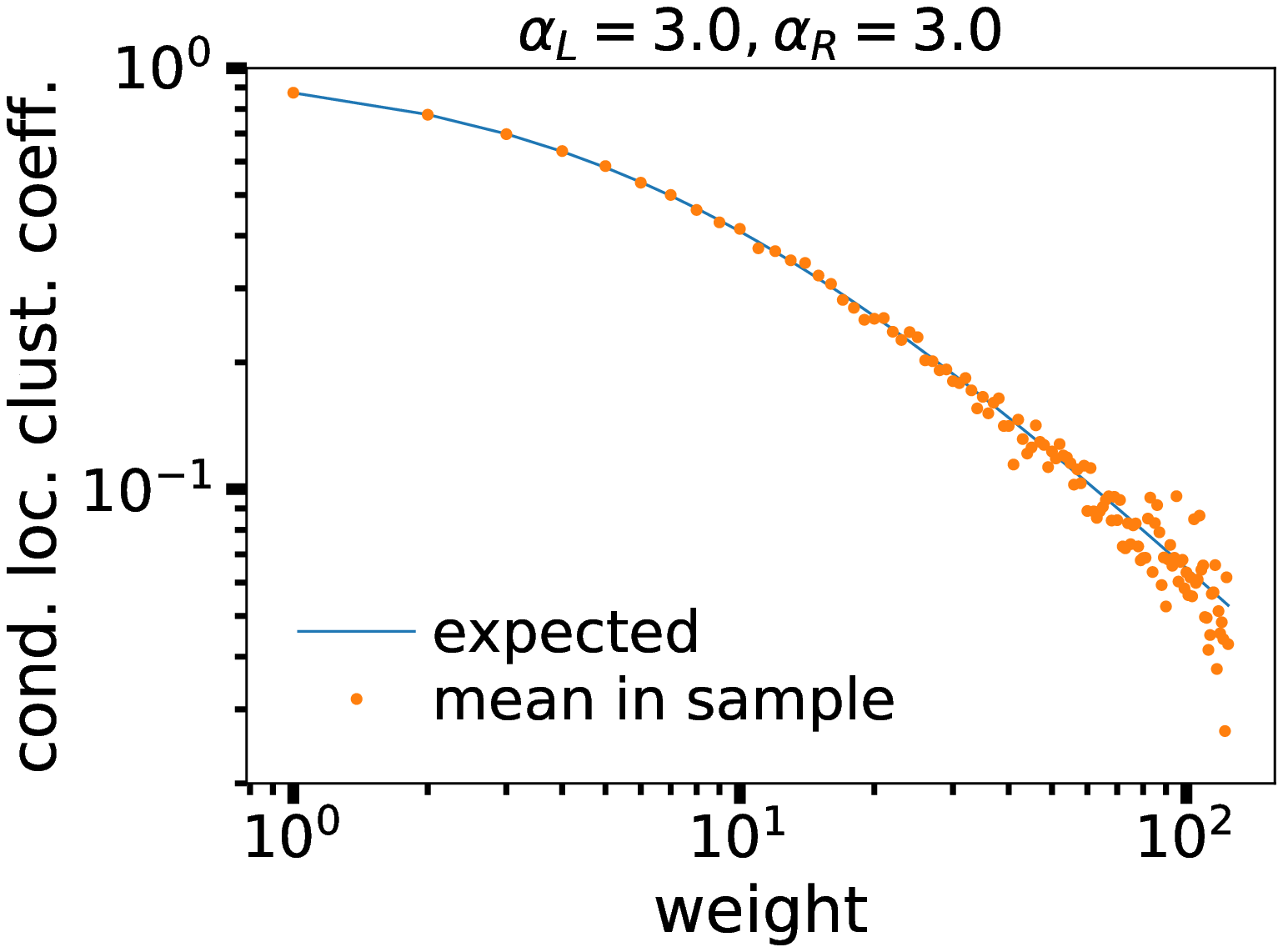}
\includegraphics[width=0.49\textwidth]{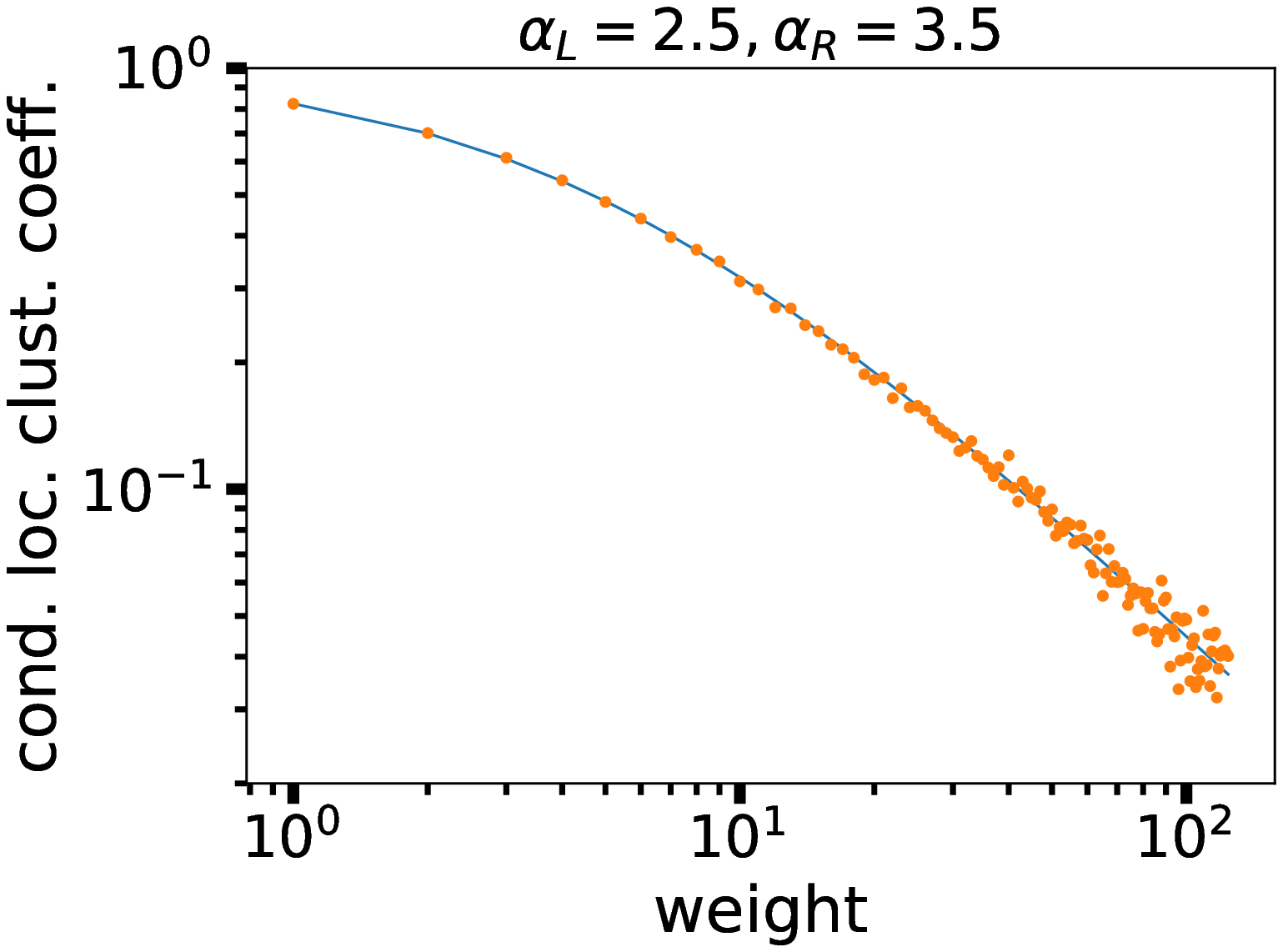}
\caption{Conditional local clustering coefficient distribution on simulated graphs as a function of node weight $w_u$,
where left and right node weights are sampled from a discrete power law distribution with decay rates $\alpha_L$ and $\alpha_R$.
The dots are the mean conditional local clustering coefficients for all nodes with that weight,
and the curve is the prediction from \Cref{thm:lccf_wt}.}
\label{fig:local_cluster}
\end{figure}

\Cref{fig:local_cluster} shows the mean conditional local clustering coefficient of a projected graph as a function of node weights $w_u$
for networks where $n_L = n_R =$ 10,000,000 and weights drawn from discrete power-law distributions with different decay parameters. 
We cap the maximum value of the weights at $n_L^{0.3}$, which corresponds to $\delta = 0.2$ in \cref{assm:WeightSeqMoments}.
The empirical clustering is close to the expected value from \cref{thm:lccf_wt}.

We can also analyze the global clustering coefficient (also called the \emph{transitivity}) of the projected graph.
The following theorem says that the global clustering tends to a constant bounded away from 0.
\begin{theorem}   \label{thm:gccf}
If \Cref{assm:WeightSeqMoments} is satisfied with $\delta > \frac{1}{10}$,
then conditioned on $S_L$ and $S_R$, we have in the projected graph that
\[
C_G = \frac{1}{1 + \frac{M_{R2} ^2}{M_{R3} M_{R1}} \cdot \frac{M_{L2}}{M_{L1}}}
 + o(1).
\]
\end{theorem}

\begin{proof}
Let $\cW$ be the set of wedges in $G$ and $\cT$ be the set of triangles. 
We first show that the global clustering coefficient is always well-defined, i.e. $\prob{|\cW| \geq 1} \geq 1 - \exp(-O(n_R))$. 
We show that with high probability, some node on the right partition has degree at least 3. 
This implies that a triangle exists in the graph and therefore a wedge exists. 
For any given node $v$ on the right, its expected degree is $w_v$ by \Cref{thm:degree_bip} 
and the degrees follow a Poisson distribution. 
By standard concentration bounds~\cite{cannone17}, $\prob{d_b(u) \leq 2} \leq \exp(-w_v/4)$ 
for $w_v$ larger than 2 (in particular, this probability is less than $1/3$ when $w_v > 7$). 
Thus, given the mild assumption that the weights have finite support for weights larger than $7$, 
the probability that there exists at least 1 triangle is $1 - \left(\frac{1}{3}\right)^{O(n_R)}$.

Next, we note that the probabilities computed in \Cref{lem:wedge_close_prob} remain unchanged 
when conditioned on the fact that at least one wedge exists. 
Let $E$ be the event that some wedge $(u, u_0, u_1)$ closes into a triangle (with $u$ as the centre of the wedge).
\[
\prob{E \cap |\cW| \geq 1} \geq \prob{E} - (1 - \prob{|\cW| \geq 1}) = \prob{E} - \left(\frac{1}{3}\right)^{O(n_R)},
\]
Consequently, 
\[
\prob{E} - \left(\frac{1}{3}\right)^{O(n_R)} \leq \prob{E \given |\cW| \geq 1} \leq \prob{E} + \left(\frac{1}{3}\right)^{O(n_R)}.
\]
Finally, $\prob{E} = \Omega(n_R^{O(1)})$ for any of the events we previously considered, so the exponentially small deviation does not produce any additional error in our results.

For any node $u$, the probability that a random wedge has center $u$ is proportional to the number of wedges centred at $u$. By our reasoning above, we can assume at least 1 wedge exists, so these probabilities sum to 1. By \Cref{lem:wedge_close_prob}, we have:
\begin{equation*}
\prob{\text{$u$ is the center node}} = \frac{\sum_{b,c \in L} \left( 1 +  \frac{M_{R1} M_{R3}}{M_{R2}^2} \cdot \frac{1}{w_u} \right)  \cdot p_{u b} p_{u c}}{\sum_{a,b,c \in L} \left( 1 +  \frac{M_{R1} M_{R3}}{M_{R2}^2} \cdot \frac{1}{w_a} \right)p_{a b} p_{a c}} + o(1).
\end{equation*}
Putting everything together,
\begin{equation*}
C_G = \sum_{u \in L} \prob{(u, u_1, u_2) \in \cT \given (u, u_1, u_2) \in \cW} \cdot \prob{\text{$u$ is the center}} 
= \frac{1}{1 + \frac{M_{R2} ^2}{M_{R3} M_{R1}} \cdot \frac{M_{L2}}{M_{L1}}} + o(1),
\end{equation*}
where the probability is taken over all $u_1, u_2 \in L$ and the second equality uses \Cref{lem:wedge_close_prob} for the probability that $(u, u_1, u_2) \in \cT$.
\end{proof}

\begin{figure}[t]
\hskip 20pt
\includegraphics[width=0.46\textwidth]{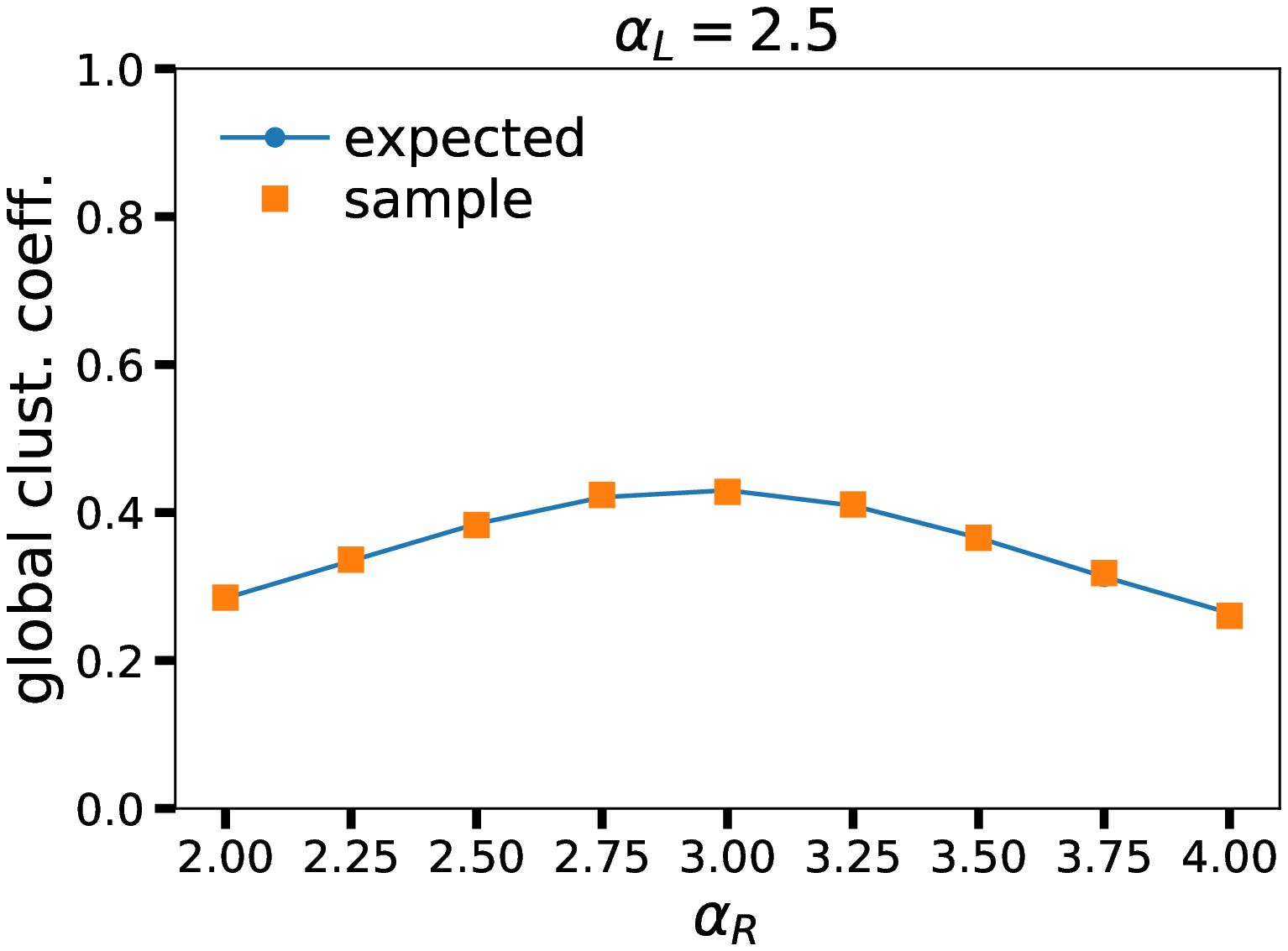}
\hfill
\includegraphics[width=0.46\textwidth]{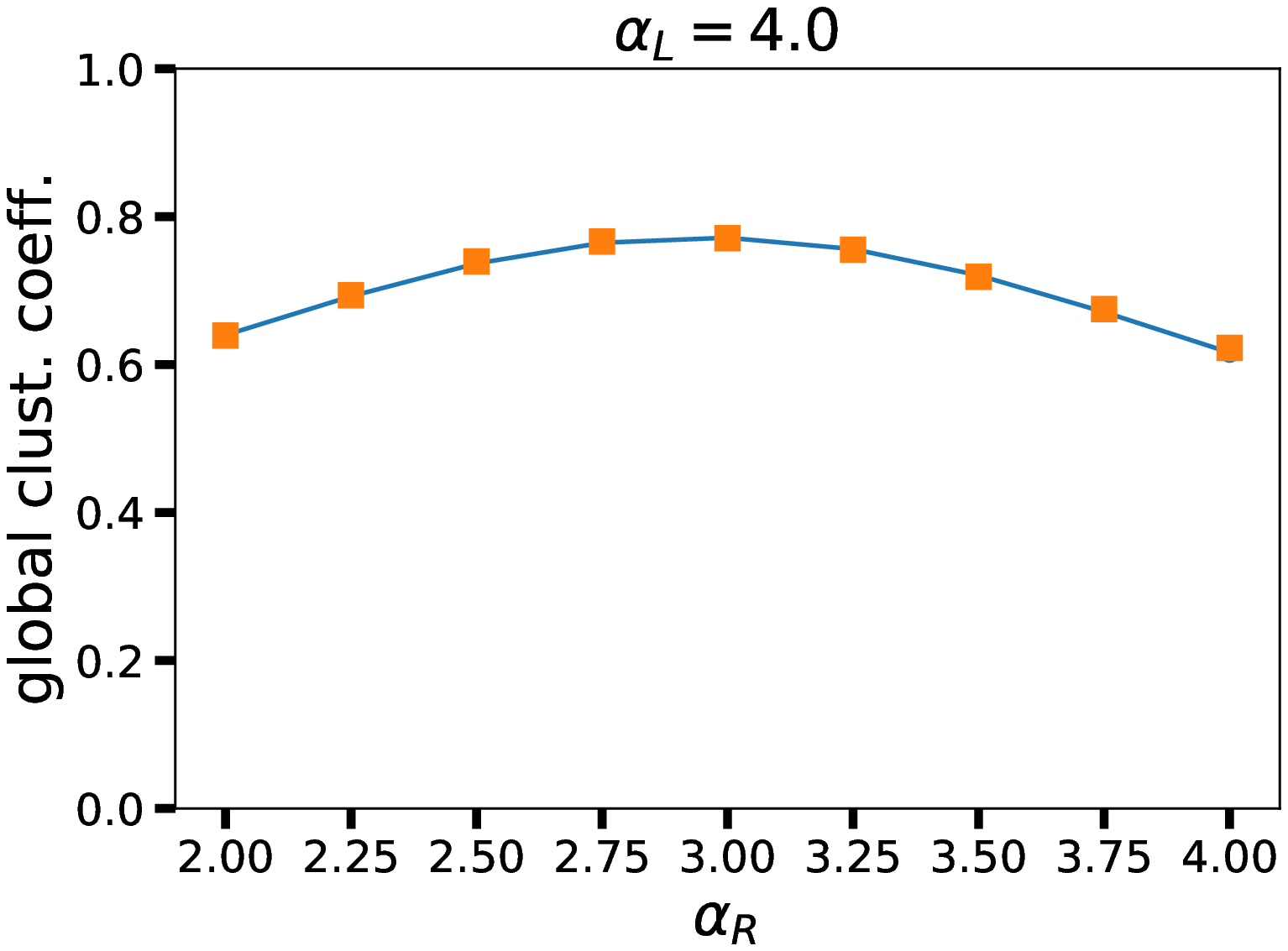}
\hskip 20pt
\caption{Expected (via \cref{thm:gccf}) and sampled global clustering coefficients on simulated graphs with discrete power law weight distributions on the left and right nodes with decay rights $\alpha_L$ and $\alpha_R$.
The samples are close to the expected value.
\label{fig:global_cluster}}
\end{figure}

\Cref{fig:global_cluster} shows the expected (computed from \cref{thm:gccf}) 
and actual global clustering coefficient of the projected graph 
with $n_L = n_R =$ 1,000,000.
The weights are drawn from a discrete power law distribution with 
fixed decay rate $\alpha_L = 2.5$ or $4.0$ on the left nodes,
varying decay rate $\alpha_R$ on the right nodes, and $w_{\max} = n_L^{0.5}$.
The sampled global clustering coefficients are close to the expectation
at all parameter values.

Finally, we investigate the local closure coefficient $H(u)$.
Analysis under the configuration model predicts that $H(u)$ should be proportional
to the node degree, while empirical analysis demonstrates a much slower increasing trend
versus degree, or even a constant relationship in a coauthorship network that is directedly 
generated from the bipartite graph projection~\cite{yin2019local}.
The following result theoretically justify this phenomenon, showing the the expected
value of local closure coefficient is independent from node weight
\begin{theorem}\label{thm:closure}
If \Cref{assm:WeightSeqMoments} is satisfied with $\delta > \frac{1}{10}$,
then conditioned on $S_L$ and $S_R$ we have, in the projected graph,
\[
H(u) = \frac{1}{1 + \frac{M_{R2} ^2}{M_{R3} M_{R1}} \cdot \frac{M_{L2}}{M_{L1}}}
 + o(1)
\]
as $n_R \to \infty$, \ie, the expected closure coefficient is asymptotically independent of node weight.
\end{theorem}
\begin{proof}
By \Cref{thm:lccf_wt}, the probability that a length-2 path $(u, v, w)$ closes into a triangle only depends on its center node $v$. Since the closure coefficient is measured from the head node $u$, the probability that any wedge is closed is independent of $u$ and thus the same across every node in the graph. This implies that the local closure coefficient is equal to the global closure coefficient, which in turn is equal to the global clustering coefficient.
\end{proof}

\begin{figure}[t]
\includegraphics[width=0.49\textwidth]{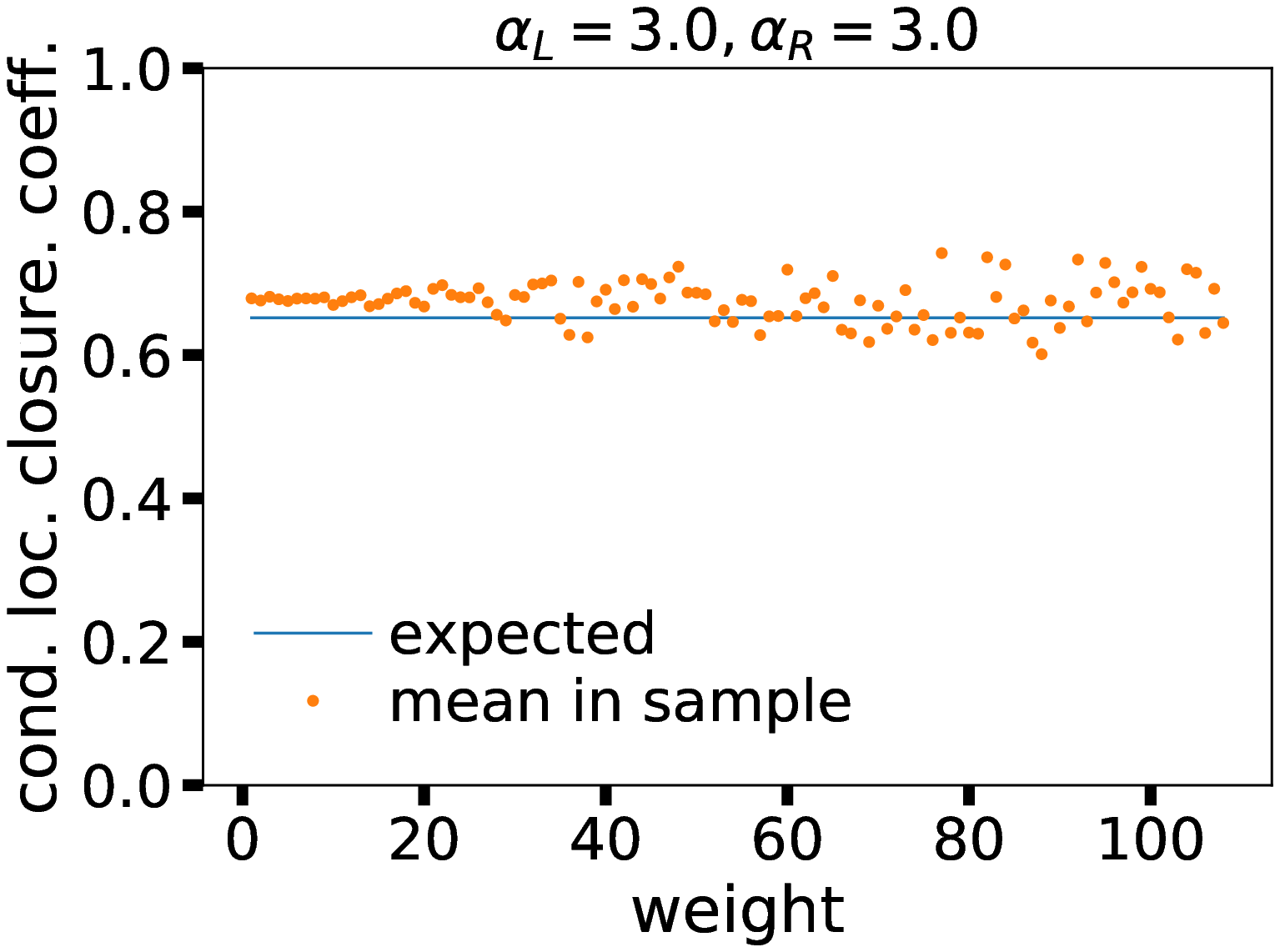}
\includegraphics[width=0.49\textwidth]{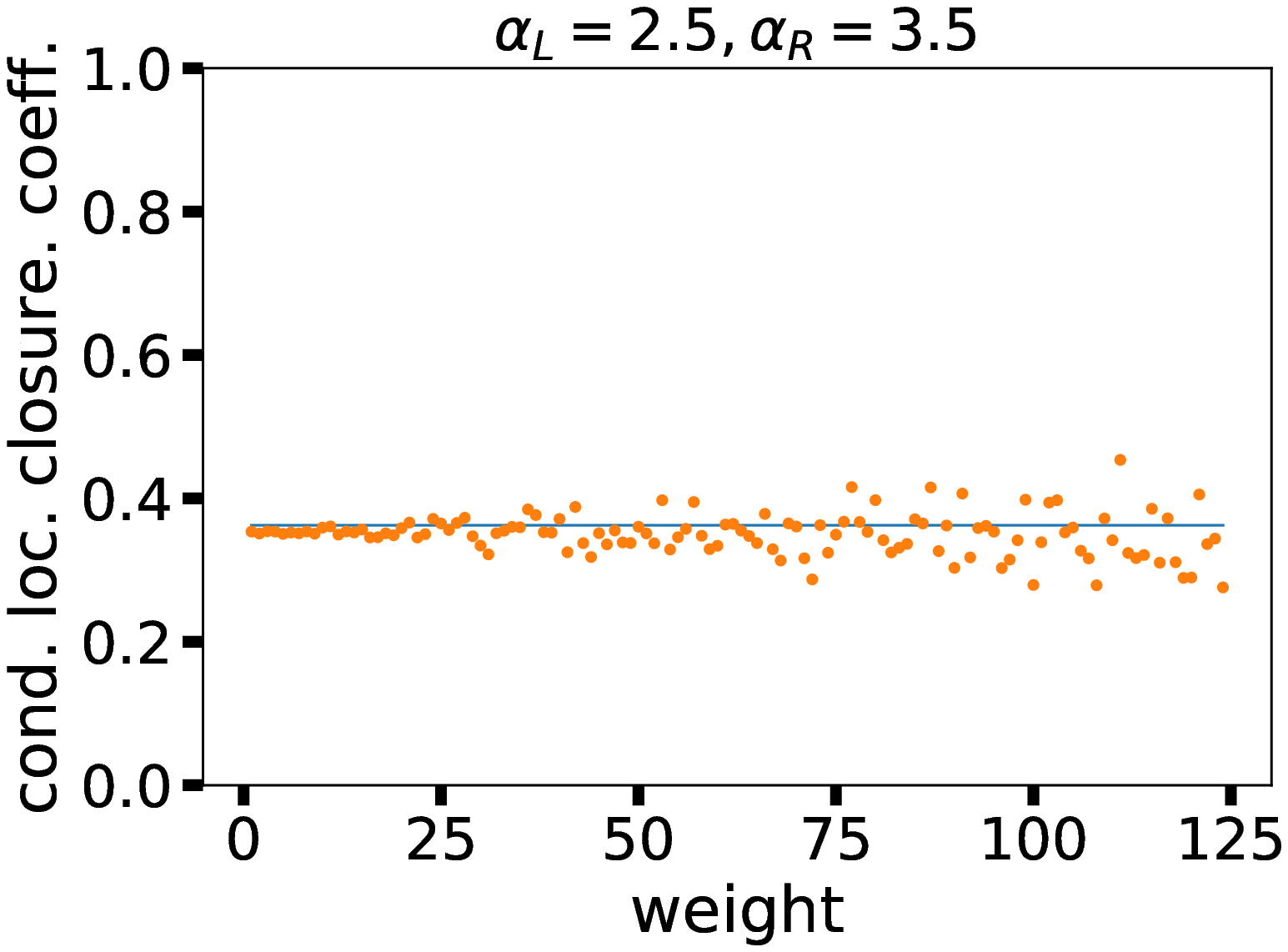}
\caption{Conditional local closure coefficient distribution on simulated graphs as a function of node weight $w_u$,
where left and right node weights are sampled from a discrete power law distribution with decay rates $\alpha_L$ and $\alpha_R$.
The dots are the mean conditional local closure coefficients for all nodes with that weight,
and the flat curve is the prediction from \cref{thm:closure}.
Weights with fewer than 5 nodes were omitted. \label{fig:local_closure}
}
\end{figure}

\Cref{fig:local_closure} shows the local closure coefficient of the projected graph as a function of node weights $w_u$,
using the same random graphs as for the clustering coefficient in \cref{fig:local_cluster}.
We observe that the mean local conditional closure coefficient is independent of the node weight in the samples,
which verifies \Cref{thm:closure}.

\begin{remark}
One can strengthen the error bounds in \cref{thm:lccf_wt,thm:gccf,thm:closure} by assuming $\delta > 1/6$. 
In particular, instead of an additive $o(1)$ error term, the error terms are a multiplicative $1 + o(1)$ factor. 
For example, the global clustering  coefficient in \Cref{thm:gccf} would be 
\[
C_G = \frac{1}{1 + \frac{M_{R2} ^2}{M_{R3} M_{R1}} \cdot \frac{M_{L2}}{M_{L1}}}(1+o(1)).
\]
\end{remark}


\section{Fast sampling and counting}\label{sec:fast_sampling}
We develop a fast sampling algorithm for
graphs with degrees following discrete power law (Zipfian) distributions, which
we use in all of our experiments.
One naive way to implement our model is to simply iterate over all $O(n_L n_R)$
potential edges, and generate a random sample for each edge. For large graphs
however, this quadratic scaling is too costly.
In contrast, our algorithm has running time linear in the number of sampled edges rather than
the product of the left and right partition sizes. This speedup is enabled by
the discrete power law distributions, which allows us to group nodes with the same weight.
The overall procedure is in \cref{alg:FastGeneration}. 

Suppose that we have two discrete power law distributions $D_L$ and $D_R$, with
$n_L\expect{D_L} = n_R\expect{D_R}$ and decay parameters $\alpha_L$ and
$\alpha_R$.  We begin by first sampling the node weights $w_u \in \mathbb{N}$
according to the specified distributions.  We then group together nodes on each
side of the bipartite graph by their weight.  With high probability, the number
of groups will be small (\cref{lem:NumGroups}).  Thus, instead of iterating over
all $O(n_L n_R)$ pairs of potential edges, we can iterate over all pairs of
groups between the left and right partition. Within each group, edges between
nodes of the group occur with a fixed probability. Hence, the number of edges
within the group follows a binomial distribution.  The final step simply
generates the number of edges $e_g$ we need from each group, and then draws that
many edges from the node pairs within that group, which can be done in linear
time.

\begin{algorithm}[t]
\caption{Fast sampling of a Chung-Lu bipartite graph with discrete power-law weights.}
\label{alg:FastGeneration}
\begin{algorithmic}
\STATE{\textbf{Input:} positive integers $n_L$, $n_R$, and degree distributions $D_L$ and $D_R$}
\STATE{\textbf{Output:} a bipartite graph $G$ following degree distributions $D_L$ and $D_R$}
\STATE{$L \gets \{1, 2, \ldots, n\}$, $R \gets \{n+1, n+2, \dots, n+n_R\}$}
\STATE{$W_L \gets \{ w_u \,|\, w_u \sim D_L\}$, $W_R \gets \{ w_u \,|\, w_u \sim D_R\}$}
\STATE{$G \gets \textrm{an empty graph with node set } L \sqcup R$}
\FOR{each unique value $(w_l, w_r) \in W_L \times W_R$}
	\STATE{$V_L \gets \{u \in L \,|\, w_u = w_l \}$, $V_R \gets \{u \in R \,|\, w_u = w_r \}$}
	\STATE{$m \gets |V_L||V_R|$, $p = \frac{w_l w_r}{n_R \mu_R}$}
	\STATE{$e_g \sim \text{Binomial}(m, p)$}
	\STATE{draw $e_g$ uniformly from $V_L \times V_R$ without replacement and add them to $G$}
\ENDFOR
\RETURN $G$
\end{algorithmic}
\end{algorithm}

\hide{
\begin{proof}
The probability of any given sample taking one of the top $t$ most probable values in $D$ is 
\begin{align}
p_t = \frac{\sum_{n=1}^t \frac{1}{n^\alpha}}{\sum_{n=1}^N \frac{1}{n^\alpha}}
= 1 - \frac{\sum_{n=t+1}^N \frac{1}{n^\alpha}}{\sum_{n=1}^N \frac{1}{n^\alpha}} 
& \geq 1 - \frac{\int_{t}^{N} x^{-\alpha}\textrm{d}x}{\int_1^{N+1} x^{-\alpha}\textrm{d}x} \\
& = 1 - \frac{N^{-\alpha + 1} - t^{-\alpha + 1}}{(N + 1)^{-\alpha + 1} - 1} 
 = 1 - O(t^{-\alpha+1}).
\end{align}
With probability $1 - O(st^{-\alpha+1})$, any collection of $s$ samples from $D$ will come from the first $t$ values. When $t = O(s^{\frac{1}{\alpha-1}})$ the probability is arbitrarily close to 1 for large $N$.
\end{proof}
}

\begin{theorem}
Let $\mu_L = \expect{D_L}$ and $\mu_R = \expect{D_R}$. The expected running time of \cref{alg:FastGeneration} 
is $O\left(n_L^{1 / (\alpha_L-1)} n_R^{1 / (\alpha_R-1)} + \mu_L n_L\right)$. For $\alpha_L, \alpha_R > 3$, the latter term dominates and the algorithm is asymptotically optimal since the second term is the expected number of edges.
\end{theorem}

\begin{proof}
By \Cref{lem:NumGroups}, the $\expect{|W_L|}$ and $\expect{|W_R|}$ are $O\left(n_L^{1 / (\alpha_L-1)}\right)$ and $O\left(n_R^{1 / (\alpha_R-1)}\right)$. 
Thus, the number of unique pairs $(w_u, w_v)$ iterated over in the for loop of \cref{alg:FastGeneration} is $O\left(n_L^{1 / (\alpha_L-1)} n_R^{1 / (\alpha_R-1)}\right)$ in expectation. Aside from the time taken to draw $e_g$ edges, each group takes constant time to process. The expected number of edges added over all the groups is $\sum_{u\in L, v\in R} \frac{w_u w_v}{n_R M_{1R}} = n_L M_{1L}$. This tends to $O(n_L \mu_L)$ in expectation. Hence the total running time is upper bounded by $O\left(n_L^{1 / (\alpha_L-1)} n_R^{1 / (\alpha_R-1)} + \mu_L n_L\right)$. 

Following \Cref{rem:equal-degrees}, we may assume without loss of generality that $\mu_L n_L = \mu_R n_R$. By the AM-GM inequality, $\frac{\mu_L n_L + \mu_R n_R}{2} \geq \sqrt{n_L n_R \mu_L \mu_R}$. For $\alpha_L, \alpha_R \geq 3$, the latter term dominates and the runtime is bounded by the expected number of generated edges. Since the output size is at least $O(\mu_L n_L)$, \Cref{alg:FastGeneration} is asymtotically optimal when $\alpha_L, \alpha_R \geq 3$.
\end{proof}

Next, we analyze the complexity of computing the graph projection along with several of the network statistics we have considered.
\begin{lemma}
Let $D_L$ and $D_R$ with decay parameters $\alpha_L$ and $\alpha_R$ be the weight distributions.
In expectation, the running time for computing all local clustering and closure coefficients and the global clustering coefficient is $O\left(n_L^{1/\min(\alpha_L, \alpha_R-1)} \frac{n_L^2 M_{1L}^2 M_{R2}}{n_R M_{1R}^2}\right)$. Under the normalization in \Cref{rem:equal-degrees}, this is equal to $O\left(n_L^{1 / (\alpha_L-1)} n_R^{1 / (\alpha_R-1)} + \mu_L n_L + \mu_R n_R\right)$.  For $\alpha_L, \alpha_R > 3$, the algorithm is asymptotically optimal, since the second term is the expected number of edges.
\end{lemma}

\begin{proof}
To compute the projected graph, we can simply iterate over all nodes $u$ in the right partition. For each pair of nodes in $N(u)$, we connect the nodes with an edge in the projected graph. Summed over all nodes in the right partition, we add $\sum_{v\in R} \left(\frac{n_L M_{1L}}{n_R M_{1R}} w_v\right)^2$ edges in the projected graph on expectation. Hence both the expected time to compute the projection and the expected number of edges in the projection is upper bounded by $O\left(\frac{n_L^2 M_{1L}^2 M_{R2}}{n_R M_{1R}^2}\right)$.

To compute the local clustering and closure coefficients, as well as the global
clustering coefficient, it is sufficient to have the degree and triangle
participation counts of each node.  The degrees are immediately available from
the projected graph, and we can list all triangles in $O(mn^{1/\alpha})$ time,
where $m$ is the number of edges in the projection, and $\alpha$ is the power
law parameter of the projection~\cite{Latapy08}. By \Cref{thm:deg_proj_distri}
and our reasoning above, $m = O\left(\frac{n_L^2 M_{1L}^2 M_{R2}}{n_R M_{1R}^2}\right)$ and $\alpha = \min(\alpha_L, \alpha_R-1)$.

By \Cref{rem:equal-degrees}, it's convenient to interpret $D_L$ and $D_R$ as the degree distributions of the left and right partitions respectively. In these cases, we have the equality $n_L \expect{D_L} = n_R \expect{D_R}$. With this equality, our results above simplify. The running time of Algorithm~\ref{alg:FastGeneration} can be restated as $O\left(n_L^{1 / (\alpha_L-1)} n_R^{1 / (\alpha_R-1)} + \mu_L n_L + \mu_R n_R\right)$. Thus for $\alpha_L, \alpha_R > 3$, the latter terms dominate (by the AM-GM inequality) and the running time is asymtotically optimal, since it is bounded by the expected number of generated edges.
\end{proof}


\section{Numerical experiments}
\label{sec:experiment}

In this section, we use our model in conjunction with several datasets. 
We find that much of the empirical clustering behavior in real-world projections can be accounted for by our bipartite project model.
All algorithms and simulations were implemented in C++, and all experiments were executed on a dual-core Intel i7-7500U 2.7 GHz CPU with 16 GB of RAM. 
Code and data are available at \url{https://gitlab.com/paul.liu.ubc/bipartite-generation-model}.

We analyze 11 bipartite network datasets (\cref{tbl:statistics}).
For the weight sequences $S_L$ and $S_R$, we use the degrees from the data.
We also compare with a version of the random intersection model~\cite{bloznelis2013degree,godehardt2003two},
where the weight sequence of the left nodes comes from the data.
For each dataset, we estimated power-law decay parameters for the degree distribution of the left and right partition (\Cref{sec:power-law-stats}).

\Cref{tbl:clustering_experiments} shows clustering and closure coefficients --- mean local clustering (i.e., average clustering coefficient), global clustering (equal to global closure), and mean local closure (i.e., average closure coefficient) ---
from (1) the data, (2) the projected graph produced by our model, and (3) the graph produced by the random intersection model. 
When computing the coefficients, we ignore any node that has an undefined coefficient, and we report the empirical
(i.e., non-conditional) variants defined in \Cref{sec:preliminaries}.

\begin{table}[t]
\caption{Description and summary statistics of real world datasets. \label{tbl:statistics}}
\begin{center}
\begin{tabular}{r l l l p{5.4cm}}
\toprule
dataset & $\lvert L \rvert$ & $\lvert R \rvert $ & $\lvert E_b \rvert$ & projection description \\
\midrule
actors-movies~\cite{barabasi1999emergence} & 384K & 128K & 1.47M & \footnotesize{actors in the same movie} \\
amazon-products-pages~\cite{leskovec2007dynamics} & 721K & 549K & 2.34M & \footnotesize{products displayed on the same page on \texttt{amazon.com}} \\
classes-drugs~\cite{Benson-2018-simplicial} & 1.16K & 49.7K & 156K & \footnotesize{FDA NDC classification codes describing the same drug} \\
condmat-authors-papers~\cite{newman2001structure} & 16.7K & 22.0K & 58.6K & \footnotesize{academics co-authoring a paper on the Condensed Matter arXiv} \\
directors-boards~\cite{Seierstad-2011-boards} & 204 & 1.01K & 1.13K & \footnotesize{directors on the boards of the same Norwegian company} \\
diseases-genes~\cite{Goh-2007-diseasome} & 516 & 1.42K & 3.93K & \footnotesize{diseases associated with the same gene} \\
genes-diseases~\cite{Goh-2007-diseasome} & 1.42K & 516 & 3.93K & \footnotesize{genes associated with the same disease} \\
mathsx-tags-questions~\cite{Benson-2018-simplicial} & 1.63K & 822K & 1.80M  & \footnotesize{tags applied to the same question on \texttt{math.stackexchange.com}} \\
mo-questions-users~\cite{Veldt-2020-local} & 73.9K & 5.45K & 132K & \footnotesize{questions answered by the same user} \\
so-users-threads~\cite{Benson-2018-simplicial} & 2.68M & 11.3M & 25.6M & \footnotesize{users posting on the same question thread on \texttt{stackoverflow.com}} \\
walmart-items-trips~\cite{Amburg-2020-categorical} & 88.9K & 69.9K & 460K & \footnotesize{items co-purchased in a shopping trip} \\
\bottomrule
\end{tabular}
\end{center}
\end{table}

\begin{table}[tb]
\caption{Clustering and closure coefficients in real-world data and in random projections following our model
and the random intersection (RI) model. 
Variances are on the order of 0.001. 
A large amount clustering is simply explained by the degree distribution and projection.
}
\label{tbl:clustering_experiments}
\begin{center}
\begin{tabular}{r ccc ccc ccc}
\toprule
                 &  \multicolumn{3}{c}{mean clust.\ coeff.}                     &    \multicolumn{3}{c}{global clust.\ coeff.}    &   \multicolumn{3}{c}{mean closure coeff.}              \\ 
                 \cmidrule(lr){2-4}                 \cmidrule(lr){5-7}                 \cmidrule(lr){8-10}
dataset                        & data   & ours & RI        & data     & ours & RI     & data     & ours & RI \\ \midrule
actors-movies          & 0.78 & 0.63 & 0.58 & 0.17  & 0.07 & 0.04 & 0.20 & 0.04 & 0.03 \\
amazon-products-pages  & 0.74  & 0.52 & 0.53 & 0.20 & 0.08 & 0.08 & 0.29 & 0.09 & 0.09 \\
classes-drugs          & 0.83& 0.79 & 0.78 & 0.50  & 0.50 & 0.49  & 0.40 & 0.24 & 0.23 \\
condmat-authors-papers & 0.74 & 0.50 & 0.50 & 0.36  & 0.12 & 0.11 & 0.35 & 0.10 & 0.10 \\
directors-boards       & 0.45 & 0.28 & 0.34 & 0.39 & 0.21 & 0.23 & 0.27 & 0.17 & 0.19 \\
diseases-genes         & 0.82 & 0.46 & 0.36 & 0.63 & 0.31 & 0.19 & 0.52 & 0.21 & 0.14 \\
genes-diseases         & 0.86 & 0.65 & 0.57 & 0.66 & 0.37 & 0.23 & 0.54  & 0.24 & 0.19 \\
mathsx-tags-questions  & 0.63 & 0.79 & 0.80 & 0.33 & 0.46 & 0.47  & 0.17 & 0.25 & 0.27 \\
mo-questions-users     & 0.86 & 0.78 & 0.64 & 0.63 & 0.45 & 0.19 & 0.37 & 0.24 & 0.19 \\
so-users-threads       & 0.40 & 0.45 & 0.46 & 0.02 & 0.01 & 0.01 & 0.00 & 0.01 & 0.01 \\
walmart-items-trips    & 0.63 & 0.55   & 0.52 & 0.05 & 0.04  & 0.04 & 0.07 & 0.02 & 0.02 \\
\bottomrule
\end{tabular}
\end{center}
\end{table}

In all but one dataset, our model has mean local clustering that is closer to the data than the
random intersection model. 
This remains true regardless of whether our model has more clustering (e.g., \emph{mathsx-tags-questions}) 
or less clustering (e.g., \emph{actors-movies}) compared to the data.
The one exception is the \emph{directors-boards} dataset, where the random intersection model accounts for more clustering than our model.
In an absolute sense, a large amount of the mean clustering is created by the projection.

To further highlight how much is explained by our model, \Cref{fig:cc-vs-degree} 
shows the local clustering coefficient as a function of degree in the data and in a sample from the models.
We find that the empirical characteristics of the clustering coefficient as a function of degree are largely
explained by the projection, suggesting that there is little innate local clustering behavior beyond what the projection 
from the degree distribution already provides.

\begin{figure}[tb]
\begin{center}
\includegraphics[width=0.47\textwidth]{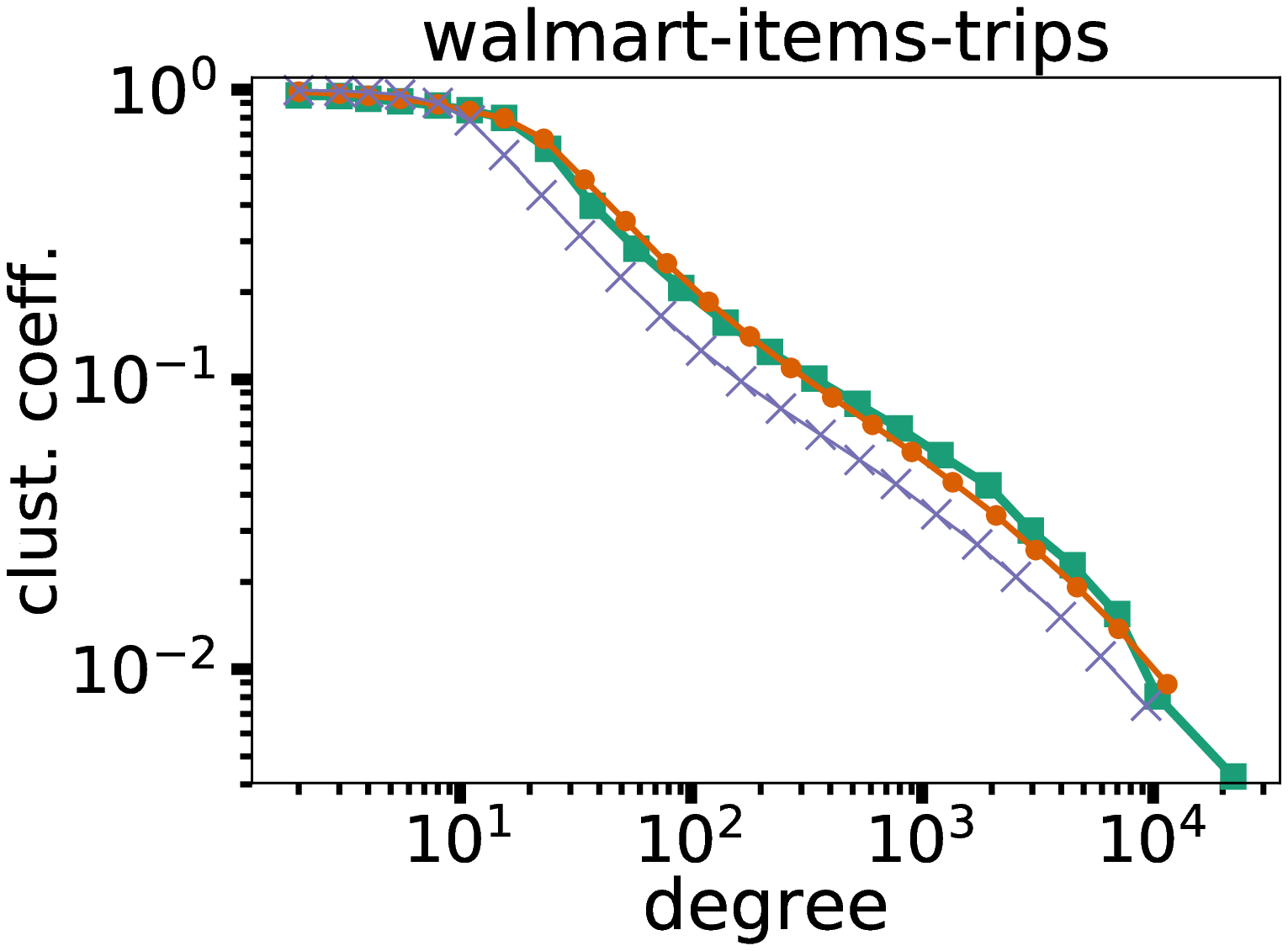}
\includegraphics[width=0.47\textwidth]{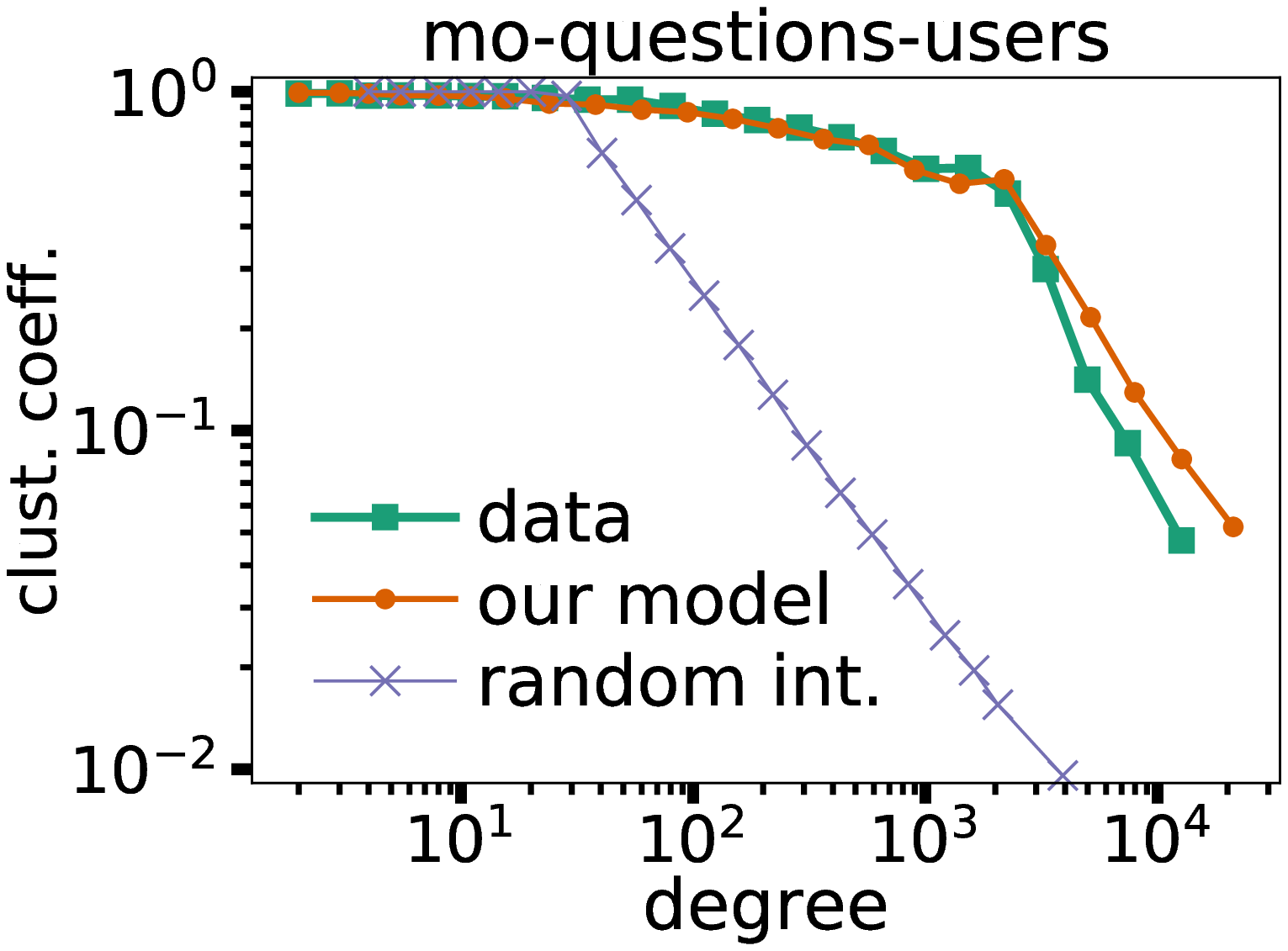}
\end{center}
\caption{Local clustering coefficient as a function of degree on the \emph{walmart-items-trips} (left) and \emph{mo-questions-users} (right) datasets. The \textcolor{mygreen}{green}, \textcolor{myorange}{orange}, and \textcolor{myblue}{blue} lines represent the clustering coefficients from the real projected graph, the projected graph produced by our model, and the projected graph produced by the random intersection model respectively. Much of the empirical local clustering behavior can be explained by the projection.
\label{fig:cc-vs-degree}}
\end{figure}

\begin{figure}[tb]
\begin{center}
\includegraphics[width=0.47\textwidth]{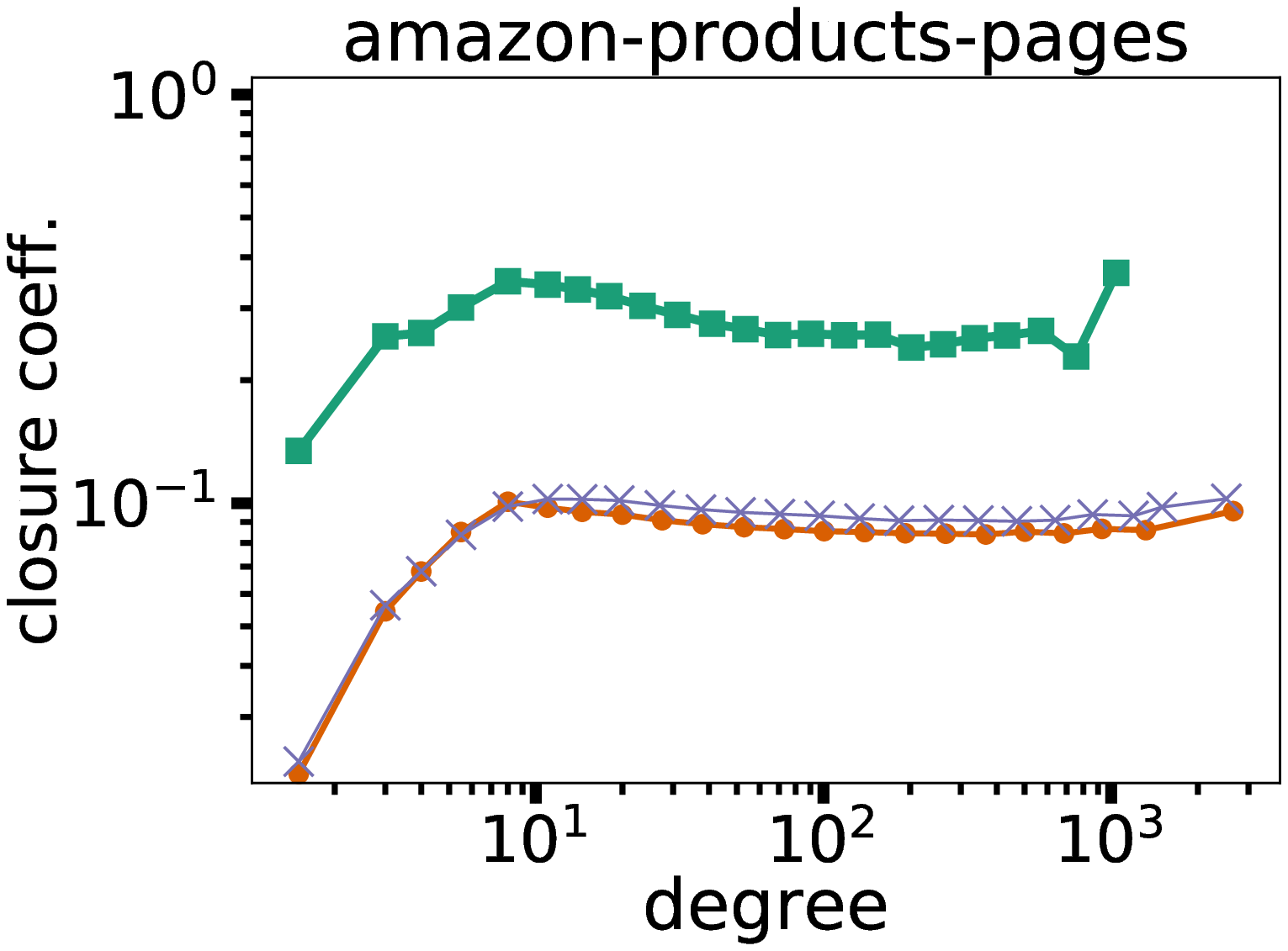}
\includegraphics[width=0.47\textwidth]{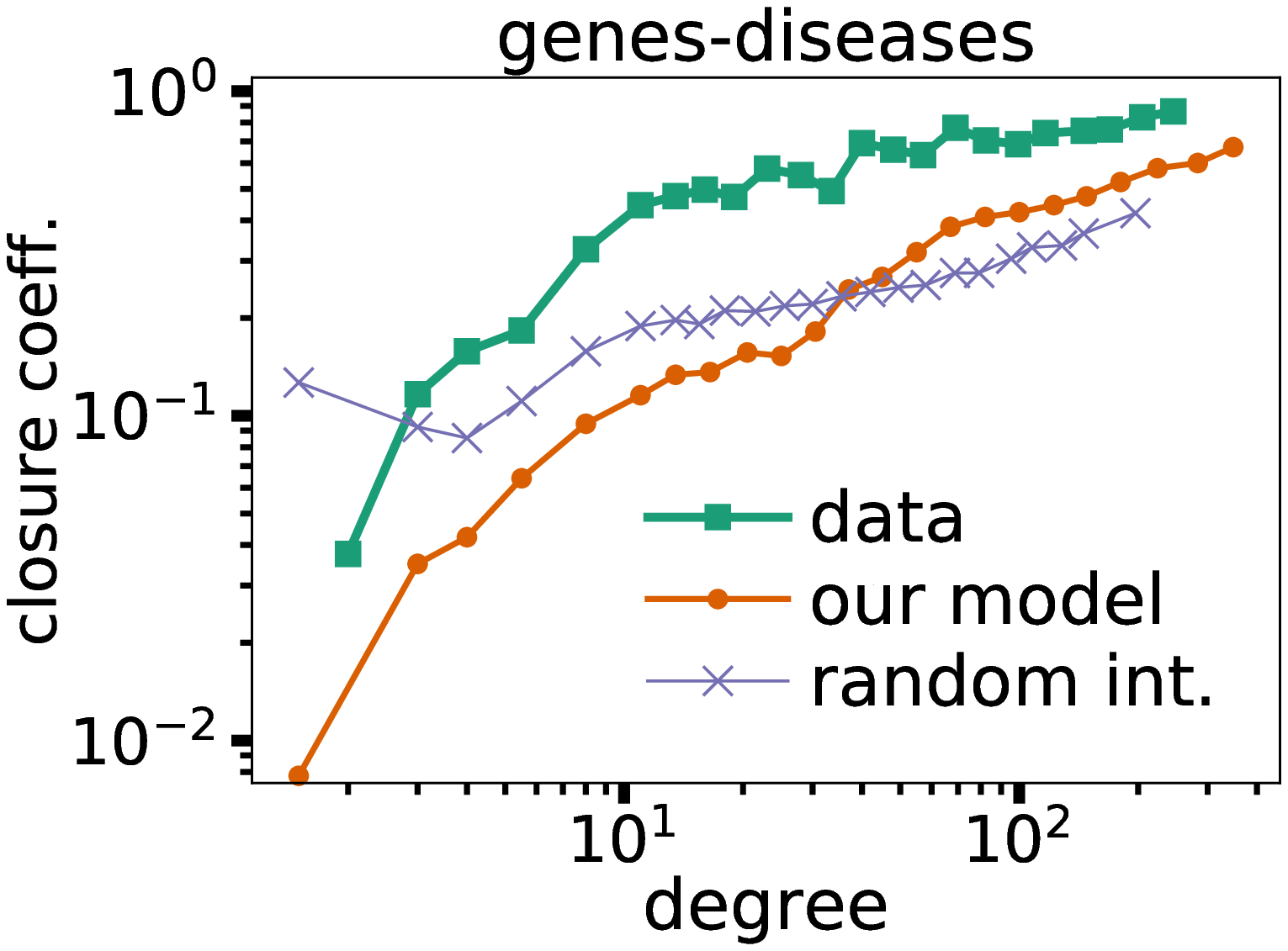}
\end{center}
\caption{Local closure coefficient as a function of degree on the \emph{walmart-items-trips} (left) and \emph{genes-diseases} (right) datasets. The \textcolor{mygreen}{green}, \textcolor{myorange}{orange}, and \textcolor{myblue}{blue} lines represent the clustering coefficients from the real projected graph, the projected graph produced by our model, and the projected graph produced by the random intersection model respectively. \label{fig:clo-vs-degree}}
\end{figure}

In some datasets, the global clustering coefficient is essentially the same as in our model (\emph{classes-drugs}, \emph{walmart-items-trips}).
However, there are several cases where our model and the random intersection model have a factor of two less global clustering
(\emph{actors-movies}, \emph{amazon-products-pages}, \emph{diseases-genes}).
This suggests that there is global transitivity in these networks that goes beyond what we would expect from a random projection.
Overall, the relative difference between the data and the model is larger for the global clustering coefficient than for the local clustering coefficient.
We emphasize that our model is not designed to match these empirical properties.
Instead, we are interested in how much clustering one can expect from a model
that only accounts for the bipartite degree distributions and the projection step.

Finally, the random graphs have non-trivial mean closure coefficients, but they tend to be smaller compared to the data,
with the exception of \emph{mathsx-tags-questions}.
Similar to the local clustering coefficient, we plot the local closure coefficient as a function of degree
for two datasets (\emph{amazon-products-pages} and \emph{genes-diseases}; \cref{fig:clo-vs-degree}).
For \emph{amazon-products-pages}, we see the flat closure coefficient as one might expect from \cref{thm:closure},
although the data has more closure at baseline.
This is likely explained by the fact that two products tend to appear on the same pages, reducing the number
of length-2 paths in the data, whereas bipartite connections are made at random in the model.
With the \emph{genes-diseases} dataset, the random models capture an increase in closure
as a function of degree that is also seen in the data.
In this case, the model parameters do not satisfy the assumptions of \cref{thm:closure}, but the general empirical
behavior is still seen in our random projection model.


\section{Conclusion}

We have analyzed a simple bipartite ``Chung-Lu style" model that captures
some common properties of real-world networks.
The simplicity of our model enables theoretical analysis of properties of the projected graph, giving analytical formulae for
graph statistics such as clustering coefficients, closure coefficients, and the
expected degree distribution. We also pair our model with a fast optimal graph
generation algorithm, which is provably optimal for certain input
distributions. Empirically, we find that a substantial amount of clustering and closure behavior
in real-world networks is explained by sampling from our model with the same bipartite degree
distribution. However, global clustering is often larger than predicted by the projection model.

\section*{Acknowledgments}
This research was supported by
NSF Award DMS-1830274,
ARO Award W911NF19-1-0057,
ARO MURI,
and JPMorgan Chase \& Co.
We thank Johan Ugander for pointing us to the literature on random intersection graphs.


\bibliographystyle{siamplain}
\bibliography{references}

\appendix


\renewcommand{\thesubsection}{\Alph{subsection}}

\section{Connection between conditional probability and empirical clustering}
\label{sec:conditional-to-standard}

Here, we show that the conditional probability formulation for clustering is exactly a weighted average of the standard empirical clustering coefficient 
for the power-law type distributions our model explores. We demonstrate this below for the local clustering coefficient. The case for local closure is similar.

Fix a node $u$ and suppose we generate a graph $G_i$ under our random graph model. Let $W_i$ and $T_i$ be the number of wedges and triangles at node $u$ in the projected graph $G_i$. The empirical clustering coefficient $\tilde{C}_i(u)$ is equal to $T_i / W_i$. Weighting each sample $\tilde{C}_i$ by $W_i$, the weighted clustering coefficient is $\frac{\sum_{i=1}^s W_i \tilde{C}_i(u)}{\sum_{i=1}^s W_i} = \frac{\frac{1}{s} \sum_{i=1}^s T_i}{\frac{1}{s}\sum_{i=1}^s W_i}$. As the number of samples $s$ (i.e. the size of the graph) approaches infinity, both the numerator and denominator approaches their expectations since each sample is independent. Computing this expectation, we see that it is exactly the value of $C(u)$ computed in \Cref{thm:lccf_wt}. 

In the case of the global closure coefficient, a similar argument shows that we actually have equality between the conditional and non-conditional definition (in the limit that the size of the graph goes to infinity).  

\section{Power-law statistics in real-world bipartite networks}
\label{sec:power-law-stats}
In many of our datasets, we find that power-law degree distributions are a reasonable
approximation for the left and right sides of the bipartite network (\cref{tbl:power-law}).

\begin{table}[h]
\caption{Estimated power law (PL) exponents of the left and right degree distributions in
the bipartite graph datasets in \cref{tbl:clustering_experiments}
(an exponent of $\alpha$ corresponds to a distribution decay $\propto k^{-\alpha}$). 
Parameters were fit using the {\tt powerlaw} pyton package~\cite{alstott2014powerlaw}.
We also report the Kolmogorov-Smirnov statistic $D$ between the fit model and the data.
}
\label{tbl:power-law}
\begin{center}
\begin{tabular}{r c c c c}
\toprule
dataset                & left PL exponent                              & D     & right PL exponent                             & D     \\ \midrule
actors-movies          & 1.862 $\pm$ 0.002 & 0.025  & 5.066 $\pm$ 0.080   & 0.019 \\
amazon-products-pages  & 3.426 $\pm$ 0.028 & 0.009 & 1.530 $\pm$ 0.001 & 0.310 \\
classes-drugs          & 2.179 $\pm$ 0.088  & 0.055 & 2.528 $\pm$ 0.007 & 0.060 \\
condmat-authors-papers & 3.495 $\pm$ 0.085   & 0.039 & 7.739 $\pm$ 0.829    & 0.017 \\
directors-boards       & 4.799 $\pm$ 0.400    & 0.055 & 5.390 $\pm$ 1.007    & 0.085 \\
diseases-genes         & 3.105 $\pm$ 0.190   & 0.043 & 3.120 $\pm$ 0.138   & 0.053 \\
genes-diseases         & 3.120 $\pm$ 0.138   & 0.053 & 3.105 $\pm$ 0.190   & 0.043 \\
mathsx-tags-questions  & 1.835 $\pm$ 0.048  & 0.048 & 5.909 $\pm$ 0.015  & 0.012 \\
mo-questions-users     & 2.842 $\pm$ 0.061  & 0.017 & 1.642 $\pm$ 0.008 & 0.031 \\
so-users-threads       & 2.407 $\pm$ 0.009  & 0.011 & 7.263 $\pm$ 0.149    & 0.015 \\
walmart-items-trips    & 2.586 $\pm$ 0.039   & 0.014 & 2.217 $\pm$ 0.005   & 0.108  \\
\bottomrule
\end{tabular}
\end{center}
\end{table}

\section{Additional proofs}\label{sec:proofs}
\subsection*{Proof of \Cref{lem:wedge_prob}}
Let $A_i$ denote the event that $(u, u_i) \in E$ for $i = 1, 2$.
We want to compute the probability of $A_1 \cap A_2$. We first decompose the probability as follows:
\begin{equation}   \label{Eq:wedge_prob_proj}
\begin{array}{rcl}
\prob{A_1 \cap A_2}
  = \prob{A_1} + \prob{A_2} - \prob{A_1 \cup A_2 }
  = \prob{A_1} + \prob{A_2} + \prob{\bar{A}_1 \cap \bar{A}_2} - 1. 
\end{array}
\end{equation}

The probability that events $A_i$ occur is given by \Cref{thm:density_proj}, so we compute the probability of $\bar{A}_1 \cap \bar{A}_2$, which is the
event that $u$ is not connected to either $u_1$ or $u_2$ in the projected graph.
This happens if and only if, in the bipartite graph,
for every $v \in R$, we have that
(i) $u$ is not connected to $v$, or
(ii) both $u_1$ and $u_2$ are not connected to $v$.
For now, let $v$ be a fixed node on the right.
Conditioning on $w_v$ and using the fact that edge formations in the bipartite graph are independent, the probability is
\begin{eqnarray*} \textstyle
1 - \frac{w_u w_v}{n_R M_{R1}} + \frac{w_u w_v}{n_R M_{R1}} \left(1 - \frac{w_{u_1} w_v}{n_R M_{R1}} \right) \left(1 - \frac{w_{u_2} w_v}{n_R M_{R1}} \right)
=1 -\frac{w_u(w_{u_1}+w_{u_2})w_v ^2}{n_R ^2 M_{R1}^2}+ \frac{w_u w_{u_1}w_{u_2}w_v ^3}{n_R ^3 M_{R1}^3}.
\end{eqnarray*}
Therefore, we have
\begin{eqnarray*}
\log(\prob{\bar A_1 \cap \bar A_2}) 
&=& \textstyle \sum_{v \in R}\log\left(1 -\frac{w_u(w_{u_1}+w_{u_2})w_v ^2}{n_R ^2 M_{R1}^2}+ \frac{w_u w_{u_1}w_{u_2}w_v ^3}{n_R ^3 M_{R1}^3} \right)\\
&=& \textstyle \sum_{v \in R} \left[-\frac{w_u(w_{u_1}+w_{u_2})w_v ^2}{n_R ^2 M_{R1}^2}+ \frac{w_u w_{u_1}w_{u_2}w_v ^3}{n_R ^3 M_{R1}^3} -\frac{w_u^2(w_{u_1}+w_{u_2})^2w_v ^4}{2n_R ^4 M_{R1}^4} \cdot (1+O(n_R^{-2\delta})) \right]\\
&=& \textstyle -p_{u u_1} - p_{u u_2} + \frac{M_{R1}M_{R3}}{M_{R2} ^2}  p_{u u_1} p_{u u_2} \frac{1}{w_u} - \frac{M_{R4}}{2n_R M_{R2}^2}(p_{u u_1} + p_{u u_2})^2 \cdot (1+O(n_R^{-2\delta})).
\end{eqnarray*}
Consequently, 
\begin{eqnarray*}
\prob{\bar A_1 \cap \bar A_2}
&=&  \textstyle 1 - p_{u u_1} - p_{u u_2} + \frac{p_{u u_1}^2}{2} + \frac{p_{u u_2}^2}{2} + p_{u u_1}p_{u u_2}  - \frac{p_{u u_1}^3 + p_{u u_2}^3}{6} (1+O(n_R^{-2\delta})) + o(p_{u u_1} p_{u u_2})
\\ & & \textstyle \quad + \frac{M_{R1}M_{R3}}{M_{R2} ^2}  p_{u u_1} p_{u u_2}\frac{1}{w_u} \cdot (1+O(n^{-2\delta})) - \frac{M_{R4}}{2n_R M_{R2}^2}(p_{u u_1}^2 + p_{u u_2}^2)(1+O(n_R^{-2\delta})).
\end{eqnarray*}
Combining everything, the probability of wedge formation is
\begin{eqnarray*} 
 \prob{A_1 \cap A_2} &=& \left(\frac{M_{R1}M_{R3}}{M_{R2} ^2} \cdot \frac{1}{w_u} + 1\right) p_{u u_1} p_{u u_2} \cdot (1 + O(n_R^{-2\delta}) + o(1))\\
 & & + \left(\frac{p_{u u_1} + p_{u u_2}}{6} + \frac{M_{R4}}{2n_R M_{R2}^2} \right) (p_{u u_1}^2 + p_{u u_2}^2) \cdot O(n_R^{-2\delta})\\
&=& \left(\frac{M_{R1}M_{R3}}{M_{R2} ^2} \cdot \frac{1}{w_u} + 1\right) p_{u u_1} p_{u u_2} \cdot \left(1 + O(n_R^{-2\delta}) + o(1) + \left(\frac{w_{u_1}}{w_{u_2}} + \frac{w_{u_2}}{w_{u_1}}\right) O(n_R^{-4\delta}) \right)
\\&=& \left(\frac{M_{R1}M_{R3}}{M_{R2} ^2} \cdot \frac{1}{w_u} + 1\right) p_{u u_1} p_{u u_2} \cdot \left(1 + O(n_R^{-2\delta}) + o(1) + O(n_R ^{1/2 - 5\delta}) \right)
\end{eqnarray*}
where the last equality is due to $w_{u_i} \in [1, n^{1/2 - \delta}]$ and their ratio is bounded by $n^{1/2 - \delta}$.
Since $\delta > 1/10$, the proof is complete.

\subsection*{Proof of \Cref{lem:triangle_prob}}
For nodes $u, u_1, u_2$ to form a triangle, one of two cases must happen. The
first is the case that all three nodes connect to a same node in the right
partition. If the first case does not happen, then each pair $(u,u_1)$,
$(u,u_2)$, $(u_1,u_2)$ have a different common neighbor in the bipartite graph,
forming a length-6 cycles.  Now we analyze these two cases separately.

In the first case, there exists a node $v \in R$ such that the three nodes $u, u_1, u_2$ 
are connected to $v$. For any specific node $v \in R$, the probability is 
$\frac{w_u w_{u_1} w_{u_2}}{n_R ^3 M_{R1}^3} \cdot w_v^3$, and thus 
\begin{eqnarray*}
& \prob{\exists v \in R \textit{ s.t. } (u, v), (u_1, v), (u_2, v) \in E_b}
  = 1 - \prod_{v \in R} \left(1 - \frac{w_u w_{u_1} w_{u_2}}{n_R ^3 M_{R1}^3} \cdot w_v^3\right) \\
  &= 1 - \exp\left(-\frac{w_u w_{u_1} w_{u_2}}{n_R ^3 M_{R1}^3} \cdot \sum_{v \in R} w_v ^3 \cdot (1 + O(n_R^{-6\delta}))  \right)
  = p_{u u_1} p_{u u_2} \cdot \frac{M_{R1}M_{R3}}{M_{R2}^2}\cdot \frac{1}{w_u} \cdot (1 + O(n_R^{-3\delta})).
\end{eqnarray*}

In the second case, $u, u_1, u_2$ are pairwise connected through a different node on the right separately, forming a 6-cycle. For any node triple $v_1, v_2, v_3$, the probability is
\begin{eqnarray*}
  \prob{\text{$(u, v_1, u_1, v_2, u_2, v_3)$ forms a 6-cycle}}
 &=& \frac{w_u^2 w_{u_1}^2 w_{u_2}^2}{n_R ^6 M_{R1} ^6} \cdot w_{v_1}^2 w_{v_2}^2 w_{v_3}^2.
\end{eqnarray*}
Therefore, the total probability of the second case is
\begin{align*}
&  \prob{\exists \text{ a 6-cycle containing $u, u_1, u_2$}} \leq \sum_{\substack{v_1, v_2, v_3 \in R \\ v_1 \neq v_2 \neq v_3}} \frac{w_u^2 w_{u_1}^2 w_{u_2}^2}{n_R ^6 M_{R1} ^6} \cdot w_{v_1}^2 w_{v_2}^2 w_{v_3}^2 \\
& \leq \frac{w_u^2 w_{u_1}^2 w_{u_2}^2}{n_R ^3 M_{R1} ^6} \cdot 
\frac{\sum_{v_1 \in R} w_{v_1}^2}{n_R} \frac{\sum_{v_2 \in R} w_{v_2}^2}{n_R} \frac{\sum_{v_3 \in R} w_{v_3}^2}{n_R}
= p_{u u_1} p_{u u_2} p_{u_1 u_2} = o(p_{u u_1} p_{u u_2}).
\end{align*}
Combining the two cases completes the proof.
\hide{
Comparing the probabilities of the two cases, we have
\[
\frac{\prob{\text{case 2}}}{ \prob{\text{case 1}}} 
\leq \frac{w_u w_{u_1}w_{u_2}}{n_R} \cdot \frac{M_{R2}^3}{M_{R1}^3M_{R3}}
 = O(n_R ^{1/2 - 3\delta}) = o(1)
\]
where the last inequality is due to $\delta > 1/6$.
Therefore, the probability of second case is negligible compared to the first case,
and thus
\[
\prob{(u, u_1) , (u, u_2) , (u_1, u_2) \in E \mid S_L, S_R}
= p_{u u_1} p_{u u_2} \cdot \frac{M_{R1}M_{R3}}{M_{R2}^2}\cdot \frac{1}{w_u} \cdot (1 + o(1)).
\]
}

\end{document}